%% file: main.tex
\documentclass[conference]{IEEEtran}

\usepackage[noadjust]{cite}

\usepackage{amsmath,amssymb,amsfonts,amsthm}

\usepackage{tikz}
\usepackage{graphicx}
\usepackage{mathtools}

\usepackage[ruled, linesnumbered]{algorithm2e}
\SetKwInput{KwData}{Input}
\SetKwInput{KwResult}{Output}

\usepackage[capitalize,noabbrev]{cleveref}

\newtheorem{theorem}{Theorem}
\newtheorem{lemma}[theorem]{Lemma}
\newtheorem{corollary}[theorem]{Corollary}

\theoremstyle{definition}
\newtheorem*{remark}{Remark}

\input{macros}

\usepackage{xspace}
\usepackage{relsize}

\def\CC{C\nolinebreak[4]\hspace{-.05em}\raisebox{.4ex}{\relsize{-2}{\textbf{++}}}}
\def\pqchannel{noisy channel\xspace}
\def\pqChannel{Noisy Channel\xspace}
\makeatletter
\def\?#1{}
\def\whp{w.h.p\@ifnextchar-{.}{\@ifnextchar.{.\?}{\@ifnextchar,{.}{\@ifnextchar){.}{\@ifnextchar:{.:\?}{.\ }}}}}}
\def\Whp{W.h.p\@ifnextchar.{.\?}{\@ifnextchar,{.}{.\ }}}
\makeatother

\def\paragraph#1{\subsubsection*{#1.\expandafter\?}}

\def\dense{\medmuskip=2.0mu plus 2.0mu minus 2.0mu
\thinmuskip=2.0mu
\thickmuskip=2.0mu plus 5.0mu}

\usepackage{balance}

\begin{document}

\title{Distributed Reconstruction of Noisy Pooled Data}

\author{\IEEEauthorblockN{Max Hahn-Klimroth}%
\IEEEauthorblockA{\textit{TU Dortmund University}\\
Dortmund, Germany \\
max.hahn-klimroth@cs.tu-dortmund.de}
\and
\IEEEauthorblockN{Dominik Kaaser}%
\IEEEauthorblockA{\textit{TU Hamburg}\\
Hamburg, Germany \\
dominik.kaaser@tuhh.de}
}

\maketitle
\thispagestyle{plain}
\pagestyle{plain}

\begin{abstract}
In the pooled data problem we are given a set of $n$ agents, each of which holds a hidden state bit, either $0$ or $1$.
A querying procedure returns for a query set the sum of the states of the queried agents.
The goal is to reconstruct the states using as few queries as possible.

In this paper we consider two noise models for the pooled data problem.
In the noisy channel model, the result for each agent flips with a certain probability.
In the noisy query model, each query result is subject to random Gaussian noise.

Our results are twofold.
First, we present and analyze for both error models a simple and efficient distributed algorithm that reconstructs the initial states in a greedy fashion.
Our novel analysis pins down the range of error probabilities and distributions for which our algorithm reconstructs the exact initial states with high probability.
Secondly, we present simulation results of our algorithm and compare its performance with approximate message passing (AMP) algorithms that are conjectured to be optimal in a number of related problems.
\end{abstract}
\begin{IEEEkeywords}
Reconstruction, Pooled Data, Random Noise, Greedy Algorithm, Approximate Message Passing
\end{IEEEkeywords}

\section{Introduction}
The distributed reconstruction problem of noisy pooled data is defined as follows.
We are given a set of $n$ agents $V = \set{x_1, \dots, x_n}$ connected via some communication network.
Each agent has a hidden state bit. 
We assume that $k$ agents have bit one and $n - k$ agents have bit zero.
The goal is to identify the agents with bit one. 
To this end, we can add \emph{query nodes} to the network.
Each query node measures a certain number of agents in parallel and returns the sum of the queried agents' states.
Our task is to design the \emph{query graph} $G$ that assigns agents to queries and a distributed algorithm that reconstructs the agents' bits using as few queries as possible.
In this setting we consider two noise models: the \pqchannel model and the noisy query model.

In the \pqchannel model, we assume that the data received by query nodes is subject to random bit flips: with probability $p$, a one bit is read as zero (\emph{false negative}), and with probability $q$ a zero bit is read as one (\emph{false positive}).
This model targets a \emph{technological setting}: query nodes can be envisioned as GPUs in a GPU cluster that evaluate a neural network \cite{liang2021neural, martins_2014, NIPS2014_fb8feff2}.
Hence, the \pqchannel model describes the possibility of random bit flips in a distributed machine learning environment.
Note that this model also includes the so-called Z-channel where $q = 0$ such that only $1 \rightarrow 0$ errors occur.
The Z-channel captures the fact that in applications $q$ is often significantly smaller than $p$ \cite{constantin1979theory,zhou2013nonuniform}.

In the noisy query model, the output of query nodes is subject to Gaussian noise.
This targets the setting of a \mbox{life-sciences} laboratory:
agents can be envisioned as samples in a medical laboratory.
The query nodes are automated pipetting machines that pool the samples together and run an automated bio-medical test.
The tests return, e.g., the total concentration of a certain substance in the pool.
Due to the pooling and the testing procedure the output of query nodes is subject to random noise.
(See \cite{kong2012automatic} for a survey on state of the art of automated life-sciences laboratories.)
Observe that in both the technological and the life-sciences setting the time to per\-form a single query dominates the time to compute the reconstruction.
We therefore focus on \emph{non-adaptive} schemes that carry out all queries in parallel.

In our analysis we distinguish between the \emph{linear} and the \emph{sublinear} regime. 
Let $k$ be the number of agents with bit $1$.
In the linear regime we have $k = \zeta n$ for some $\zeta \in (0,1)$.
In the sublinear regime we have $k = n^{\theta}$ for some $\theta \in (0,1)$.
Both regimes are of practical interest.
For example, according to Heaps' Law of Epidemiology \cite{benz_2008} the early spread of a pandemic can be modeled by the sublinear regime.
Similarly, in 2019 it was estimated that there are 105,200 people living with HIV in the UK, out of which an estimated 6\% are unaware of their infection status \cite{aids}.
This corresponds roughly to a value of $\theta = 0.1$.
On the other hand, various tasks in computational biology \cite{du2000combinatorial,cao_2014,sham_2002}, traffic monitoring \cite{wang_2015} or confidential data transfer \cite{adam_1989,dinur_2003} fall into the linear regime.

While the pooled data problem is studied quite frequently (e.g., \cite{AJS_book, alaoui_2017, gebhard_2022, grebinski_2000, hahnklimroth2021near}), the existing line of research mostly analyzes the idealized setting without noise. However, there is always a positive probability of misclassifications in neural networks, and bio-medical testing procedures are known to be prone to noise.
In this more realistic, noisy setting, only few contributions are known \cite{li2022combinatorial, scarlett_2017}.
They all study information-theoretic aspects and thus assume unlimited computational power. Our main motivation, therefore, is to bridge the gap between theoretical results and practical applications.
In particular, we analyze a distributed variant of a computationally efficient but sequential algorithm proposed by Gebhard et al.\ \cite{gebhard_2022}, which has been previously analyzed only in the idealized setting without noise. For our distributed variant we give precise performance guarantees that show that the algorithm works well under realistic noise models.
To the best of our knowledge our work is the first study of an efficient algorithm for the pooled data problem under noise.

\subsection{Our Contribution}
The contribution of this paper is twofold.
First, in \cref{sec_theoretic_proofs} we give rigorous and exact asymptotic bounds on how many queries are needed such that our algorithm reconstructs the hidden bits of all agents correctly \whp\footnote{We say that a sequence of properties $\cP_1, \ldots, \cP_n$ holds with high probability (\whp), if $\lim_{n \to \infty} \Pr \bc{ \cP_n } = 1$.}. 
The formal definition of the pooling model and the noise models can be found in \cref{sec:model}, and \cref{algorithm} is described in \cref{sec:algorithm}.
Secondly, in \cref{sec:simulations} we present extensive simulation results for realistic numbers of agents $n \in \bc{10^2 \ldots 10^5}$ that evaluate those asymptotic bounds.
Additionally, we compare our findings with the performance of the \emph{approximate message passing algorithm (AMP)}.
AMP is a sequential algorithm that is conjectured to be optimal in similar problems \cite{donoho_amp, donoho_amp2}.

For the following theorems we define $\gamma = 1 - \exp \bc {-1/2}$ and denote by Z the Z-channel and by GNC the general noisy channel. 

\begin{theorem}[Noisy Channel Model] \label{thm_pq}
Let $n$ be the number of agents and assume that $m$ queries are conducted. Furthermore, let $k = n^{\theta}$ in the sublinear regime and $k = \zeta n$ in the linear regime.  \Cref{algorithm} recovers the hidden bits of all agents correctly \whp, if, for any $\eps > 0$, the following holds. 
\begin{itemize}
    \item Sublinear Regime
    \begin{itemize}
        \item[(Z)] $ m \geq (4 \gamma + \eps) \frac{ \bc{1 + \sqrt{\theta}}^2 }{1-p} k \log(n)$
        \item[(GNC)] $ m \geq (4 \gamma + \eps)  \frac{ q \bc{1 + \sqrt{\theta}}^2 }{ \bc{1-p-q}^2 } n \log(n)$
    \end{itemize}
    \item Linear Regime
    \begin{itemize}
        \item[(Z/GNC)] $m \geq (16 \gamma + \eps) \frac{q + (1 - p - q) }{(1 - p - q)^2 } \zeta n \log(n)$
    \end{itemize}
\end{itemize}
\end{theorem}
\begin{remark}
We assume that $p, q \in [0, 1)$ with $p + q < 1$ do not depend on $n$. The conducted analysis directly implies that, asymptotically, $q = o \bc{\frac{k}{n}}$ behaves exactly as $q = 0$ and $q = \omega\bc{\frac{k}{n}}$ behaves as $q > 0$ in the previous theorem. 
We furthermore observe that our results for $p=q=0$ match the results of \cite{gebhard_2022}, which have been only shown for the case without noise.
\end{remark}

In the noisy query model, we assume that every query undergoes an independent $\cN(0, \lambda^2)$ noise.
Depending on the size of $\lambda$, \cref{algorithm} can be safely applied or not.
More precisely, the algorithm's success probability undergoes a \emph{phase transition}.
\begin{theorem}[Noisy Query Model]\label{thm_noisy}
Let $n$ be the number of agents and assume that $m$ queries are conducted. 
If $\lambda^2 = o \bc{ \frac{m}{\log n}}$, \cref{algorithm} recovers the hidden bits correctly \whp if the following holds.
\begin{itemize}
    \item Sublinear: $ m \geq (4 \gamma + \eps) { \bc{1 + \sqrt{\theta}}^2 } k \log(n)$,
    \item Linear: $m \geq (16 \gamma + \eps) \zeta n \log(n)$.
\end{itemize}
If $\lambda^2 = \Omega \bc{ m }$, \cref{algorithm} fails to recover the hidden bits with positive probability for any number of queries conducted.
\end{theorem}

\subsection{Related Work}
The pooled data problem can, in its simplest variants, be traced back to work of Dorfman \cite{dorfman_1943}, Shapiro \cite{shapiro_1960}, \Erdos~and \Renyi~\cite{erdos_1963}, and Djackov~\cite{djackov_1975}. 
Well-known variants of the problem have been studied under the name of \textit{(quantitative) group testing} \cite{cao_2014,karimi_2019,karimi_2019_2,wang_2015,AJS_book, coja_spiv, feige2020quantitative, hahnklimroth2021near}, \emph{threshold} group testing \cite{Chan2013, DeMarco2020}, \textit{coin weighing} \cite{bshouty_2009,djackov_1975,shapiro_1960,soderberg_1963}, or \textit{pooled data} \cite{alaoui_2017,scarlett_2017,wang_2016}.
It has a multitude of applications, from computational biology \cite{du2000combinatorial,cao_2014,sham_2002} over traffic monitoring \cite{wang_2015} and confidential data transfer \cite{adam_1989,dinur_2003} to machine learning \cite{wang_2016,liang2021neural,martins_2014,NIPS2014_fb8feff2}. For a survey of applications, see \cite{du2000combinatorial}.

The \emph{binary group testing problem} was recently studied under noise. 
In this simple variant, queries output only the information whether at least one agent with hidden bit $1$ is contained.  Those studies compromise fundamental information-theoretic results \cite{AJS_book} and algorithmic limits \cite{scarlett_2020}
as well as practical applications \cite{zhu2020noisy}.

The variant we study is less explored. 
Even though the noiseless case is well understood w.r.t.\  information-theoretic and algorithmic aspects \cite{cao_2014,grebinski_2000, karimi_2019_2,wang_2015,AJS_book, feige2020quantitative, hahnklimroth2021near, Marco_2013, alaoui_2017}, only few results for noisy measurements are known.
These results discuss fundamental limits assuming unlimited computational power \cite{li2022combinatorial, scarlett_2017}, but no efficient algorithms are known in this setting.
Based on an efficient sequential algorithm designed for the noiseless case \cite{gebhard_2022}, we extend the line of research by a careful analysis of a distributed variant under noise.

In line with relevant recent literature on variants of the pooled data problem, agents are assigned to several queries by placing each agent independently and randomly into queries \cite{alaoui_2017,karimi_2019,lee_2015,scarlett_2017,wang_2016, hahnklimroth2021near}. We furthermore allow that an agent is included in a query multiple times, which adapts techniques used in a variety of other statistical inference problems~\cite{aldridge_2016,aco_2019,johnson_2019}.

Finally, we emphasize that we restrict ourselves to the \emph{non-adaptive} setting, in which all queries must be conducted independently and in parallel. This setting became the prevalent variant in the noiseless case recently 
\cite{aldridge_2014,scarlett_2017,zhang_2013,alaoui_2017, hahnklimroth2021near, gebhard_2022, AJS_book, wang_2016}. This reflects the fact that in most applications the time to conduct a single queries dominates the whole reconstruction time \cite{liang2021neural, martins_2014, NIPS2014_fb8feff2, kong2012automatic}.


\section{Model}
\label{sec:model}
We assume that out of $n$ agents, exactly $k$ agents have the hidden bit 1 while the remaining $n-k$ agents have the bit 0. We denote the bits by $x_1 \ldots x_n$ and let $\SIGMA \in \cbc{0,1}^n$ be the so-called \emph{ground-truth}: a vector that represents the \emph{real} (unknown) bit of each agent. As usual in reconstruction problems, we assume that $\SIGMA$ is uniformly chosen among all binary vectors of Hamming weight $k$ and length $n$.

Assume that $m$ queries are conducted in a distributed fashion and let $a_1, \ldots, a_m$ denote these queries.
Then every $a_j$ is a multi-set of agents $x_1, x_2, \ldots$.
In our algorithm we use a model where each query has the same size $\Gamma = \frac{n}{2}$, and any query picks $\Gamma$ agents from $V$ uniformly at random with replacement.
This is in line with recently studied variants of the pooled data problem without noise \cite{alaoui_2017, feige2020quantitative, gebhard_2022}.

As in the recently studied noiseless variant of the problem, we let $d = d(n): \NN \to \RR^+$ be a function and set the number of query-nodes $m$ in the network to $m = d k \log(n).$

We let $\Delta_1, \ldots, \Delta_n$ denote the (random) number how often each agent is queried. Furthermore, we let $\Delta^\star_1, \ldots, \Delta^\star_n$ denote the (random) number of \emph{distinct} such queries.
Finally, we let $\hat \SIGMA \in \RR^m$ denote the random vector of query results. 

We conclude that an instance of the pooled data problem can be represented by a \emph{bipartite multi-graph}, with agents being one class and queries being the other class. An edge indicates that an agent is queried by a given query node. See \cref{fig:figure_pooling} for an example.
In line with this representation we denote by $\partial x$ the multi-set of neighbors in the graph (multi-edges counted multiple times) and by $\partial^\star x$ the set of distinct neighbors. Thus, $\abs{\partial x_i} = \Delta_i$ and $\abs{\partial^\star x_i} = \Delta^\star_i$ for all agents and $\abs{ \partial a_j } = \Gamma$ for all queries $a_j$.

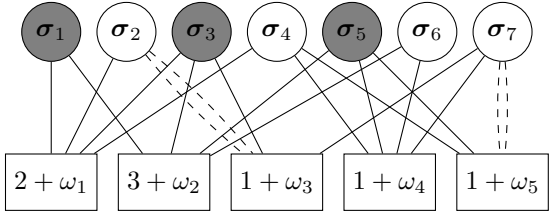
\begin{figure}[t]
\centering
\begin{tikzpicture}[scale=1]
\node[circle, draw, minimum width=0.75cm, fill=black!50] (x0) at (0, 0) {$\SIGMA_1$};
\node[circle, draw, minimum width=0.75cm] (x1) at (1,0) {$\SIGMA_2$};
\node[circle, draw, minimum width=0.75cm, fill=black!50] (x2) at (2, 0) {$\SIGMA_3$};
\node[circle, draw, minimum width=0.75cm] (x3) at (3, 0) {$\SIGMA_4$};
\node[circle, draw, minimum width=0.75cm, fill=black!50] (x4) at (4, 0) {$\SIGMA_5$}; 
\node[circle, draw, minimum width=0.75cm] (x5) at (5, 0) {$\SIGMA_6$};
\node[circle, draw, minimum width=0.75cm] (x6) at (6, 0) {$\SIGMA_7$};

\node[rectangle, draw, minimum width=0.75cm, minimum height=0.75cm] (a1) at (0, -2.0) {$2 + \omega_1$};
\node[rectangle, draw, minimum width=0.75cm, minimum height=0.75cm] (a2) at (1.5, -2.0) {$3 + \omega_2$};
\node[rectangle, draw, minimum width=0.75cm, minimum height=0.75cm] (a3) at (3, -2.0) {$1 + \omega_3$};
\node[rectangle, draw, minimum width=0.75cm, minimum height=0.75cm] (a4) at (4.5, -2.0) {$1 + \omega_4$};
\node[rectangle, draw, minimum width=0.75cm, minimum height=0.75cm] (a5) at (6, -2.0) {$1 + \omega_5$};

\path[draw] (x0) -- (a1);
\path[draw] (x0) -- (a2);
\path[draw] (x1) -- (a1);
\path[-] (x1) edge [bend left=5,dashed] (a3);
\path[-] (x1) edge [bend right=5,dashed] (a3);
\path[draw] (x2) -- (a1);
\path[draw] (x2) -- (a3);
\path[draw] (x2) -- (a2);
\path[draw] (x3) -- (a1);
\path[draw] (x3) -- (a4);
\path[draw] (x3) -- (a5);
\path[draw] (x4) -- (a2);
\path[draw] (x4) -- (a5);
\path[draw] (x4) -- (a4);
\path[draw] (x5) -- (a2);
\path[draw] (x5) -- (a4);
\path[draw] (x6) -- (a3);
\path[draw] (x6) -- (a4);
\path[-] (x6) edge [bend left=5,dashed] (a5);
\path[-] (x6) edge [bend right=5,dashed] (a5);
\end{tikzpicture}
\caption{A small example with $n = 7$ agents and ground-truth $\SIGMA = (1,0,1,0,1,0,0)$. The multi-graph indicates which agents are queried by which query node. The goal is to reconstruct $\SIGMA$ given only $\G$ and the query results. The results are subject to noise $\omega_1, \ldots, \omega_5$.}
\label{fig:figure_pooling}
\end{figure}

\subsection{\pqChannel Model}
The \pqchannel represents typical noisy observations in technological settings like machine learning in which communication is subject to random noise. The level of noise can be dependent on the ground truth value of the queried hidden bit. With respect to the pooled data problem, the \pqchannel is defined as follows. In one query, we find exactly $\Gamma$ single agents that are queried (probably more than once). For each of those single edges in the bipartite graph (see \cref{fig:figure_pooling}), we receive a different signal depending on the ground-truth of the agent. More precisely, if agent $x$ is queried, then the query result is increased by a random variable $S(x)$ where
\begin{align*}
    S(x) = \begin{cases} 1 & \text{with probability } q \text{ if } \SIGMA_j = 0 \\
    1 & \text{with probability } 1 - p \text{ if } \SIGMA_j = 1 \\
    0 & \text{with probability } 1-q \text{ if } \SIGMA_j = 0 \\
    0 & \text{with probability } p \text{ if } \SIGMA_j = 1.
    \end{cases}
\end{align*}
If $q = 0$, this model is known as the \emph{binary asymmetric channel}, or \emph{Z-channel}.
We assume throughout the paper that $p + q < 1$ and $p, q \in [0, 1)$ are (known) constants.

\subsection{Noisy Query Model}
In the noisy query model, we assume that all hidden bits are measured correctly. Nevertheless, the reading procedure that evaluates the sum of the hidden bits, is exhibited to Gaussian noise. More precisely, 
\begin{align*}
    \hat \SIGMA_a = \sum_{x \in \partial a} \SIGMA_x + \cW_a,
\end{align*}
where $\cW_1, \ldots, \cW_m$ are independent Gaussians with mean $0$ and variance $\lambda^2$, thus $\cW_a \sim \cN(0, \lambda^2)$.
We emphasize that this noise model has a second interpretation: on each sample out of the $\Gamma$ probes in the pooled query $a$, we have a small Gaussian fluctuation distributed as $\cN \bc{0, \lambda^2 \Gamma^{-1}}$ independently from all other sources of noise.
This represents the inaccuracy of pipetting machines quite well.

\section{Algorithms} \label{sec:algorithm}

\paragraph{Noisy Maximum Neighborhood Algorithm}
The Maximum Neighborhood Algorithm is a sequential greedy algorithm that was recently proposed and analyzed in the noiseless case by Gebhard et al.\ \cite{gebhard_2022}. We extend their results in two ways. First, we introduce an equivalent distributed variant.
Secondly,  we analyze the performance of this algorithm under the noise models defined in \cref{sec:model}.

On an intuitive level, our algorithm works as follows. Agents interact in a classical message passing environment. Query nodes sample multi-sets of agents and measure their bits (subject to noise). Then they send their measurements to all involved agents. The agents sums up all query results they receive and sort themselves via a sorting network (see, e.g., \cite{batcher1968sorting,sortingnetworks}). Those agents with the highest sums are declared as having hidden bit one. The distributed algorithm is formally specified in \cref{algorithm}. A more detailed description is given in \cref{sec_noiseless}.

\begin{algorithm}[ht]

\DontPrintSemicolon
\SetAlgoVlined
\SetKwFor{ForParallel}{for}{do in parallel}{end}
\SetKwFor{ForAt}{at}{do}{end}
\SetKw{KwTo}{to}
\SetKw{KwWith}{with}
\SetKwFunction{Query}{query}
\SetKwProg{Initially}{I.~Perform Measurements in Parallel}{}{end}
\SetKwProg{Reconstruct}{II.~Reconstruct Bits via a Sorting Network}{}{end}
\Initially{}{
add $m$ query nodes to the network\;
\ForAt{query node $a_j$}{
    sample a set of agents $\cbc{v_1, \dots, v_\Gamma}$ u.a.r.\ with replacement from $V$\;
    measure $\hat \SIGMA_j = \sum_{i = 1}^{\Gamma} \SIGMA(v_i)  + \omega_{ij}$ \;
    \tcc*{$\hat \SIGMA_j$ {\normalfont\small is subject to noise!}}
    let $\partial^\star a_j$ be the set of \textbf{distinct} neighbors of $a_j$.\;
    send message $\hat \SIGMA_j$ to each agent in $\partial^\star a_j$ 
}
\ForAt{agent $x_i$ \KwWith incoming message $\hat \SIGMA_j$}{
    update score $\Psi_i \gets \Psi_i + \hat \SIGMA_j$\;
    update degree $\Delta^\star_i \gets \Delta^\star_i + 1$\;
    initialize permutation $\pi_i = i$\;
}
}
\Reconstruct{}{
use a sorting network on $\{x_1, \dots, x_n\}$ to sort the\; permutation $\pi$ on $[n]$ according to
    $\displaystyle \Psi_i - \Delta^\star_i{k}/{2}$\;
    \tcc{{\normalfont\small set the agents with the $k$ largest scores to 1}}
    \parbox{2.5cm}{\textbf{if} $\pi_i \leq k$ \textbf{then}} agent $x_i$ outputs $1$\;
    \parbox{2.5cm}{\textbf{else} } agent $x_i$ outputs $0$\;
}
\caption{Greedy Reconstruction}
\label{algorithm}
\end{algorithm}

\paragraph{Approximate Message Passing (AMP)}
AMP finds its roots in statistical physics and is assumed to be optimal in robust reconstruction problems \cite{donoho_amp, donoho_amp2}. Robust means that the influence of one single agent in a specific query is negligible. This is the case for the pooled data problem: we expect about $k/2$ agents with hidden bit one in a query such that a difference of $\pm 1$ is not detectable (in an asymptotic setting).
%
AMP is frequently applied to the so-called \emph{compressed sensing problem} \cite{Le_Gallo_2018}, which can be seen as the following natural generalization of the pooled data problem: each agent has a hidden opinion $\SIGMA_j \in \RR$ (rather than a discrete bit). For this problem, AMP is assumed to be optimal if the number of non-zero opinions $k$ is small compared to $n$. A non-rigorous and heuristic but intuitively convincing analysis of AMP in the noiseless linear variant of the pooled data problem is given by Alaoui et al.~\cite{alaoui_2017}. 

Formally, the algorithm is defined as follows.
As an input, it gets 
the pooling graph (\cref{fig:figure_pooling}) as an adjacency matrix $\vec A \in \NN_0^{m \times n}$ and the query result vector $\hat \SIGMA$.
Observe that $\vec A \SIGMA + \cW = \hat \SIGMA$ where $\cW$ is (unknown) noise.
AMP defines the following update rules \cite{Le_Gallo_2018} that give an estimate of $\SIGMA$:
\begin{align*}
    \sigma^{(t+1)} &= \eta_t \bc{ \vec A^T z^{(t)} + \sigma^{(t)}} \\
    z^{(t)} & = \hat \SIGMA - \vec A \sigma^{(t)} + \frac{n}{m} \frac{1}{n} \sum_{i=1}^n \eta'_{t-1} \bc{ \vec A^T z^{(t-1)} + \sigma^{(t-1)} }.
\end{align*}
Hereby, $(\eta_t)_{t \geq 0}$ is a family of possibly time-dependent functions $\eta_t: \RR \to \RR$ that are applied coordinate wise such that $\eta_t(\sigma) = \bc{\eta_t(\sigma_1), \ldots, \eta_t(\sigma_n)}$ for $\sigma \in \RR^n$.
The vector $z^{(t)}$ describes the deviation of the query results with regard to the current estimate of $\SIGMA$.
It contains the so-called \emph{Onsager term} $\frac{n}{m} \frac{1}{n} \sum_{i=1}^n \eta'_{t-1} \bc{ \vec A^T z^{(t-1)} + \sigma^{(t-1)} }$, which accounts for under-sampling effects if $k/n$ is small \cite{donoho_amp, donoho_amp2}.
The standard initialization reads $\sigma^{(0)} = (0, \ldots, 0)$.

AMP admits the following intuitive description. In every step, every query node sends out a \emph{message} to all its connected agents that contains information about how well the current estimate $\sigma^{(t)}$ is compatible with its query result. All agents then update their estimate of their hidden bit. This is iterated many rounds, which is the reason why AMP is hard to analyze rigorously.
We remark that AMP has an intuitive description in a distributed message passing environment. However, the communication overhead becomes substantial rendering (unmodified) AMP inefficient in this setting \cite{han2014distributed}.

For our simulations we use the recent implementation by Safarpour et al.\ \cite{Safarpour_2018}.

\section{Analysis} \label{sec_theoretic_proofs}
\subsection{Preliminaries}
Let $\vec G$ denote the random bipartite multi-graph described in the model section with $n$ agents $x_1, \ldots, x_n$ and $m$ queries $a_1, \ldots, a_m$. We let $\cG$ denote the $\sigma-$algebra induced by the edges of the graph. Thus, the sequences of random degrees $\Delta_1, \ldots, \Delta_n$ and $\Delta_1^\star, \ldots, \Delta_n^\star$ are $\cG-$measurable.  
Those random quantities are tightly concentrated around their means.
\begin{lemma} \label{cor_conc_delta} Let $\Delta = \Erw \brk{ \Delta_j } = m \Gamma n^{-1} = m/2$.
With probability $1 - o(n^{-1})$, we find
\begin{align*}
    \Delta - \log(n) \sqrt{\Delta} \leq \min_{j \in [n]} \Delta_i \leq \max_{j \in [n]} \Delta_i \leq \Delta + \log(n) \sqrt{\Delta}.   
\end{align*}\label{conc_delta}
\end{lemma}
\begin{proof}
Exactly $m \Gamma$ half-edges are thrown uniformly at random into the set of $n$ agents. For a single agent $x_j$, the probability to receive a specific half-edge is $n^{-1}$, independent of the choices of the different half-edges. Therefore, $\Delta_i \sim \Bin(m \Gamma, n^{-1})$ and the corollary follows from Chernoff bounds (\cref{lem_chernoff} in the appendix).
\end{proof}
\begin{lemma}[Lemma 3 of \cite{gebhard_2022}] \label{cor_delta_star}
We have, \whp, 
\begin{align*}
    \Delta^\star_i = 2 \bc{ 1 - \exp \bc{-1/2} } \Delta_i + O \bc{ \log(n) \sqrt{\Delta} }.
\end{align*}
\label{num_multiedges}
\end{lemma}

\begin{corollary} \label{cor_conc_delta_star} Let $\Delta^\star = \Erw \brk{ \Delta^\star_j } = \bc{ 1 - \exp \bc{-1/2} } m$.
With probability $1 - o(n^{-1})$, we find
{\dense
\begin{align*}
    \Delta^\star - \log^2(n) \sqrt{\Delta^\star} \leq \min_{j \in [n]} \Delta^\star_i \leq \max_{j \in [n]} \Delta^\star_i \leq \Delta^\star + \log^2(n) \sqrt{\Delta^\star}.
\end{align*}
}
 \label{conc_deltastar}  
\end{corollary}
\begin{proof}
    This is a direct consequence of \cref{cor_conc_delta} and \cref{cor_delta_star}.
\end{proof}

We denote by $\cR$ the event that the concentration properties of \cref{conc_delta,num_multiedges,conc_deltastar} hold in the random graph $\vec G$. By the previous corollaries, we have $\Pr \bc{ \cR } = 1 - o(1)$.

\subsection{Recap of the Noiseless Case} \label{sec_noiseless}
As \cref{algorithm} is a distributed variant of the \emph{maximum neighborhood algorithm} of \cite{gebhard_2022}, we recap the main idea of the algorithm in the noiseless case. If any agent sums up the distinct query results it is part of, one might hope that agents with hidden bit 1 have a larger \emph{neighborhood sum}\footnote{The expression neighborhood sum was introduced by \cite{gebhard_2022} and refers to the representation as a bipartite graph.}. 

By definition, agent $x_j$ is part of $\Delta^\star_j$ distinct queries. 
The other agents that appear in a query with $x_j$ are independent of $x_j$. In terms of the random bipartite graph, we see that the second neighborhood of $x_j$ is (almost) equally distributed between agents with hidden bit 1 and hidden bit 0. If we count the number of observed hidden bits in the second neighborhood of $x_j$ under the ground-truth $\SIGMA$, this is a binomially distributed random variable $\Xi_j$.
\begin{lemma}[Corollary 4 of \cite{gebhard_2022}]\label{cor_xi_nonoise}
The number of hidden bits with value one in the second neighborhood of agent $x_j$ is distributed as $\Xi_j$ where
\[ \Xi_j \sim \Bin \bc{ \Gamma \Delta_j^\star - \Delta_j, \frac{k - \vecone \cbc{ \SIGMA_j = 1 }}{n - 1} }.\]
\end{lemma}
\begin{proof}
There are $\Gamma \Delta^\star_j - \Delta_j$ (not necessarily distinct) agents $x_i \neq x_j$ that are part of queries in the neighborhood of $x_j$. Each of those agents is one of $k - \vecone \cbc{\SIGMA(j) = 1}$ agents with hidden bit 1 independently out of the $n-1$ remaining agents. 
\end{proof}
If there was no noise the neighborhood sum $\Psi_j$ of agent $x_j$, given the query results vector $\hat \sigma$, formally,
\[ \Psi_j(\hat \sigma) = \sum_{i=1}^m \vecone \cbc{ a_i \in \partial x_j } \hat \sigma_i, \]
is given by $ \Psi_j \sim \Xi_j + \Delta_j \vecone \cbc{\SIGMA_j = 1}$. Indeed, the hidden bit of $x_j$ increases this sum in every of the $\Delta_j$ queries by 1. Conducting more queries increases the concentration properties of $\Xi_j$. Thus, if sufficiently many queries are conducted, it is possible to separate agents with hidden bit one and hidden bit zero by the score \[ \Psi_j - \Erw \brk{ \Xi_j } \sim \Psi_j - \Delta_j^\star \frac{k}{2}\] \whp, see \cite[Theorem 1]{gebhard_2022}.

This main idea can be carried over to the noisy setting but the underlying distributions become much more involved: the queries do no longer report the hidden bits reliably.

\subsection{\pqChannel Model}
In this section we prove \cref{thm_pq} for the \pqchannel model.
The description $\Xi_j$ of the number of agents with hidden bit 1 in the second neighborhood of agent $x_j$ stays clearly correct with respect to the ground-truth $\SIGMA$. However, it does not describe the distribution of the query results anymore. Indeed, every single edge connecting a query and an agent in the underlying graph is, with a certain probability, noisy. Given the number of agents $n_j$ in the second neighborhood of $x_j$, thus $n_j = \Delta^\star_j \Gamma - \Delta_j$, we observe four different \emph{types} $t_{00}, t_{01}, t_{10}, t_{11}$ of hidden bits in the second neighborhood. Here, type $t_{ij}$ is read as \emph{having bit $i$ under $\SIGMA$ and receiving bit $j$ in the query}. Observe that, if the same agent gets queried multiple times, the noise is independent. Therefore, a single agent can appear with hidden bit 0 and hidden bit 1. We define
\begin{equation} \label{eq_pij}
\begin{aligned}
     p_{j}(0,0) &= \bc{ 1 - \frac{k - \vecone \cbc{\SIGMA_j = 1}}{n-1} } (1 - q) \\
     p_j(0,1) & = \bc{1 - \frac{k - \vecone \cbc{\SIGMA_j = 1}}{n-1}} q \\
     p_{j}(1,1) &= { \frac{k - \vecone \cbc{\SIGMA_j = 1}}{n-1} } (1 - p)  \\
    p_j(1,0) & = {\frac{k - \vecone \cbc{\SIGMA_j = 1}}{n-1}} p
    \end{aligned}
\end{equation}
with the interpretation that a single drawn random agent in the second neighborhood of $x_j$ is of type $t_{ik}$ with probability $p_j(i,k)$ independently of everything else. Indeed, $p_{j}(i,k)$ is just the product of the probability that we observe hidden bit $i$ under the ground-truth and the probability of the result under the \pqchannel.
Given $n_j$, we introduce the multinomially distributed random variable \[ \Lambda_j \sim \Mult \bc{ n_j, p_j(0,0), p_j(0,1), p_j(1,0), p_j(1,1) }.\] We interpret $\Lambda_j$ as a random vector with the four entries counting the number of occurrences of type $t_{00}, t_{01}, t_{10}$ and $t_{11}$.
Therefore, the number of observed hidden bits of value one in the second neighborhood of agent $x_j$, previously called $\Xi_j$, is distributed as follows with respect to $\SIGMA$.
\begin{lemma}
The number of observed hidden bits with value one in the second neighborhood of agent $x_j$ under the \pqchannel is distributed as $\Xi^{pq}_j$ where
\[ \Xi^{pq}_j \sim \Lambda_j(0,1) + \Lambda_j(1,1). \]
\end{lemma}
\begin{proof}
As in the proof of \cref{cor_xi_nonoise}, exactly $n_j$ (not necessarily distinct) agents are queried. In contrast to the noiseless case, the observed bit is one with probability $p_j(0,1) + p_j(1,1)$ as either bit zero was observed but switched by the channel or bit one was transferred correctly. The remainder of the proof follows the lines of \cref{cor_xi_nonoise}.
\end{proof}
Now, we see that, with respect to the ground-truth $\SIGMA$, the neighborhood sum $\Psi_j(\hat \sigma)$ is distributed as follows.
\begin{lemma}
\label{lem_neighborhoodsum_pq}
Let $\PSI_j = \Psi_j(\hat \SIGMA)$ be the neighborhood sum of agent $x_j$ with respect to the query results vector induced by the ground-truth $\SIGMA$. Then,
\begin{align*}
    \PSI_j = \Psi_j(\hat \SIGMA) \sim \Xi_j^{pq} & + \vecone \cbc{ \SIGMA_j = 1 } \Bin \bc{ \Delta_j, 1-p } \\ & + \vecone \cbc{ \SIGMA_j = 0 } \Bin \bc{ \Delta_j, q }.
\end{align*}
\end{lemma}
\begin{proof}
The second neighborhood is distributed as $\Xi_j^{pq}$ by the previous corollary. The agent $x_j$ herself increases the neighborhood sum every time the bit was read correctly if $\SIGMA_j = 1$ while an agent with hidden bit $\SIGMA_j = 0$ increases its neighborhood sum every time the bit was read falsely. The number of times that this happens, independently, is given by the binomial distribution. 
\end{proof}
It is a well known fact that $\Lambda_j(0,1)$ and $\Lambda_j(1,1)$ are negatively associated binomial random variables such that \[\Lambda_j(i,k) \sim \Bin \bc{ n_j, p_j(i,k) }.\]
Therefore, the neighborhood sum itself is, given the status of the hidden bit of an agent, a sum of three negatively associated binomial random variables and its concentration properties can be pinned down by Chernoff bounds (\cref{lem_chernoff} in the appendix). This fact is the main idea of the subsequent proofs. There is a discrepancy of the expected size of the neighborhood sum with respect to the hidden bit of agent $x_j$. As $\Delta_j$ itself is, by construction, a binomial random variable, it is tightly concentrated around its expectation $\Delta$ (see \cref{cor_conc_delta}). Indeed, given $\cR$, we find
\begin{equation}\label{eq_expected_difference}\begin{aligned}
    \MoveEqLeft \Erw  \brk{ \PSI_j \mid \SIGMA_j = 1, \cG } - \Erw \brk{ \PSI_j \mid \SIGMA_j = 0, \cG } \\ & = \Delta  (1 - p - q) + O \bc{ \log^2(n) \sqrt{ \Delta } }.
\end{aligned}
\end{equation}
Thus, we will establish conditions that guarantee that 
\begin{align}
    \label{eq_score_exp} \PSI_j - \Erw \brk{ \Xi_j^{pq} \mid \cG, \cR} > \PSI_i - \Erw \brk{ \Xi_i^{pq} \mid \cG, \cR}
\end{align}
for all $x_j$ with hidden bit one and all $x_i$ with hidden bit zero \whp.

Recall $p_j(0,1)$ and $p_j(1,1)$ from \eqref{eq_pij} and recall furthermore that $\Xi_j^{pq}$ is the sum of two binomially distributed and negatively associated random variables with $n_j = \Delta_j^\star \Gamma - \Delta_j$ trials and success rate $p_j(01)$ and $p_j(1,1)$ respectively. Then, 
\begin{equation}\label{eq_xi_expectation}
\begin{aligned} 
    \Erw \brk{ \Xi_j^{pq} \mid \cG, \cR } & = \bc{\Delta_j^\star \Gamma - \Delta_j} \bc{ p_j(0,1) + p_j(1,1) } \\
    & = (1 + o(1)) \Delta^\star \Gamma \bc{ q + \frac{k}{n} (1 - p - q) }
\end{aligned}
\end{equation}
While given the pooling graph $\vec G$, it is an easy task to calculate $\Erw \brk{ \Xi_j^{pq} \mid \cG, \cR }$, the exact value of $\Xi_j^{pq}$ is not available to any algorithm. Therefore, the ultimate goal is to conduct enough queries such that $\abs{ \Xi_j^{pq} - \Erw \brk{ \Xi_j^{pq} \mid \cG, \cR } }$ is so small that the difference $\PSI_j - \PSI_i$ in \eqref{eq_score_exp} of approximately $ \Delta_j (1 - p - q) + O \bc{ \log(n) \sqrt{\Delta_j} }$ between agents with hidden bit one and hidden bit zero is significantly larger than the difference $\Xi_j^{pq} - \Xi_i^{pq}$ \whp. In this case, the agents with hidden bit one are detectable by \cref{algorithm}.

Chernoff bounds yield {\dense
\begin{align}
    \Pr & \bc{ \abs{ \Xi_j^{pq} - \Erw \brk{ \Xi_j^{pq} \mid \cG, \cR } } > \beta \Delta_j \mid \cG, \cR } \notag \\
    & \leq 2 \Pr \bc{ \Xi_j^{pq} \geq \bc{1 + \frac{\beta \Delta_j}{\Erw \brk{ \Xi_j^{pq} \mid \cG, \cR }} } \Erw \brk{ \Xi_j^{pq} \mid \cG, \cR } \mid \cG, \cR } \notag  \\
    & \leq 2 \exp \bc{ - \frac{\beta^2 \Delta_j^2}{ 2 \Erw \brk{ \Xi_j^{pq} \mid \cG, \cR } } }  \notag \\
    & \leq 2 \exp \bc{ - (1 + o(1)) \frac{\beta^2 \Delta^2}{ 2 \Delta^\star \Gamma \bc{ q + \frac{k}{n} (1 - p - q) } } }. \label{eq_chernoff_pc_prep}
\end{align}}

\noindent Two things need to be established. First, $\PSI_j - \Erw \brk{ \Xi_j^{pq} \mid \cG, \cR }$ must overshoot a certain value for all $k$ agents with hidden bit one. Second, it needs to be undershot by all $n-k$ agents with hidden bit one. Observe that for $\alpha \in (0,1)$, 
\[ \Delta_j q + \alpha \Delta (1 - p - q) = \Delta_j (1-p) - (1 - \alpha) \Delta (1-p-q) \]
interpolates between the expected difference in the neighborhood sum for agents of different state. 

Formally, we require
\begin{align*}
    \PSI_j & - \Erw \brk{ \Xi_j^{pq} \mid \cG, \cR }  <  \Delta q + \alpha \Delta (1 - p - q)
\end{align*}
\whp for all agents with hidden bit zero and
\begin{align*}
    \PSI_j & - \Erw \brk{ \Xi_j^{pq} \mid \cG, \cR }  > \Delta (1-p) - (1 - \alpha) \Delta (1-p-q)
\end{align*}
\whp for all agents with hidden bit one. 

By \cref{eq_expected_difference,eq_chernoff_pc_prep} we have 
{\dense \allowdisplaybreaks
\begin{align}
    \notag \Pr & \big( \PSI_j - \Erw \brk{ \Xi_j^{pq} \mid \cG, \cR \mid \SIGMA_j = 0} \\
    \notag & \qquad \qquad > \Delta q + \alpha \Delta (1 - p - q) \mid \SIGMA_j = 0 \big)\\
    \notag & \leq \Pr  \big( \abs{\Xi_j^{pq} - \Erw \brk{ \Xi_j^{pq} \mid \cG, \cR } } > (\alpha + o(1)) \Delta_j (1 - p - q) \big) \\
    & \leq  2 \exp \bc{ - (1 + o(1)) \frac{\alpha^2 \bc{1-p-q}^2 \Delta^2}{ 2 \Delta^\star \Gamma \bc{ q + \frac{k}{n} (1 - p - q) } } }
    \label{eq_chernoff_pc_1}
\intertext{and}
    \notag \Pr & \big( \PSI_j - \Erw \brk{ \Xi_j^{pq} \mid \cG, \cR \mid \SIGMA_j = 1} \\
    \notag & \qquad > \Delta (1-p) - (1-\alpha) \Delta (1 - p - q) \mid \SIGMA_j = 1 \big)\\
    \notag & \leq \Pr  \big( \abs{\Xi_j^{pq} - \Erw \brk{ \Xi_j^{pq} \mid \cG, \cR } } > (1-\alpha + o(1)) \Delta_j (1 - p - q) \big) \\
    & \leq  2 \exp \bc{ - (1 + o(1)) \frac{(1-\alpha)^2 \bc{1-p-q}^2 \Delta^2}{ 2 \Delta^\star \Gamma \bc{ q + \frac{k}{n} (1 - p - q) } } }.
    \label{eq_chernoff_pc_2}
\end{align}}

\noindent To establish a union bound over all $k$ agents with hidden bit 1 and $n-k$ agents with hidden bit 0, it therefore suffices to pin down $\alpha = \alpha(p,q)$ and the number of queries conducted parametrized by $\Delta = m \Gamma / n$, such that

\begin{align} \label{eq_condition_k}
    \frac{(1-\alpha)^2 \bc{1-p-q}^2 \Delta^2}{ (2 + o(1)) \Delta^\star \Gamma \bc{ q + \frac{k}{n} (1 - p - q) } } > \log(k)
\end{align}
and
\begin{align} \label{eq_condition_n}
    \frac{\alpha^2 \bc{1-p-q}^2 \Delta^2}{ (2 + o(1)) \Delta^\star \Gamma \bc{ q + \frac{k}{n} (1 - p - q) } } > \log(n)
\end{align}
As the l.h.s.\ of \eqref{eq_condition_n} is increasing in $\alpha$ while the l.h.s.\ of \eqref{eq_condition_k} is decreasing in $\alpha$, there is exactly one $\alpha \in (0,1)$ in which both conditions equal and are therefore satisfied as weakly as possible.
Observe that for $p = q = 0$ \cref{eq_condition_k,eq_condition_n} recover the conditions of \cite{gebhard_2022}. 

Recall that $\Delta = m \Gamma / n$, $\Delta^\star = 2 (1 - \exp\bc{-1/2}) \Delta$, $\Gamma = n/2$ and $m = d k \log(n)$.
We distinguish the two cases $k = n^{\theta}$ with $\theta \in (0,1)$ and $k = \zeta n$ with $\zeta \in (0,1)$. 
\paragraph{Sublinear Case}
The condition \eqref{eq_condition_k} now reads
\begin{align} \label{eq_condition_k_sublinear}
    \frac{(1-\alpha)^2 \bc{1-p-q}^2 \Delta^2}{ (2 + o(1)) \Delta^\star \Gamma \bc{ q + \frac{k}{n} (1 - p - q) } } - \theta \log(n) > 0. 
\end{align}

The denominator's behavior changes dramatically depending on $q$, the probability to falsely read a zero as a one. We distinguish two cases.

\paragraph{Case $q = 0$ (Z-channel)}
In this case, the denominator in the conditions reads $(2 + o(1)) (1 - p) \Delta^\star \Gamma \frac{k}{n}$.
Therefore, we need to establish
\begin{align} \label{eq_condition_n_sub11}
    \frac{\alpha^2 \bc{1-p} d}{ (4 + o(1)) ( 1 - \exp( - 1/2) ) } - 1 > 0
\end{align}
and
\begin{align} \label{eq_condition_k_sub11}
    \frac{(1-\alpha)^2 \bc{1-p} d}{ (4 + o(1)) ( 1 - \exp( - 1/2) ) } - \theta > 0 .
\end{align}
We let $\gamma = (1 - \exp \bc{-1/2})$ for brevity. The conditions in \cref{eq_condition_n_sub11,eq_condition_k_sub11} are satisfied as weakly as possible if
\begin{align} \label{eq_condition_alpha_sub11}
    \bc{ \alpha^2 - (1 - \alpha)^2 } = \frac{ (4 + o(1)) \gamma (1 - \theta)}{d (1 - p)}. 
\end{align}
 Then the optimal solution is \[ \alpha = (1 + o(1))\frac{ 4 \gamma (1-\theta) + d(1-p)}{2d(1-p)}.\]
A short calculation that involves calculation of the roots of a parabola verifies that \cref{eq_condition_n_sub11} evaluated at this point yields
\begin{align*}
    d > (4 \gamma + o(1)) \frac{ \bc{1 + \sqrt{\theta}}^2 }{1-p}.
\end{align*}
(It is straightforward to check that the choices of $\alpha$ and $d$ satisfy the equation.)
Therefore, \cref{algorithm} reconstructs $\SIGMA$ correctly \whp if
\[ m \geq (1 + \eps) 4 \gamma \frac{ \bc{1 + \sqrt{\theta}}^2 }{1-p} k \log(n) \]
queries are conducted as claimed in \cref{thm_pq}. Observe that this directly extends the results by \cite{gebhard_2022}.

\paragraph{Case $q > 0$ (general noisy channel)}
The denominator in the conditions in \cref{eq_condition_k_sublinear} reads $(2 + o(1)) \Delta^\star \Gamma q$. Again, let $\gamma = (1 - \exp \bc{-1/2})$. 
We need to establish
\begin{align} \label{eq_condition_n_sub1}
    \frac{\alpha^2 \bc{1-p-q}^2 d k}{ (4 \gamma + o(1))  n q } - 1 &> 0
\intertext{and}
\label{eq_condition_k_sub1}
   \frac{(1-\alpha)^2 \bc{1-p-q}^2 d k}{ (4 \gamma + o(1)) n q } - \theta &> 0.
\end{align}
As the nominator scales in $d k$ and the denominator in $n$, the number of queries conducted (given by $d k \log(n)$) needs to be substantially larger. We therefore set \[
d = c n  k^{-1} \] for some constant $c > 0$. Then \cref{eq_condition_n_sub1,eq_condition_k_sub1} are satisfied as weakly as possible if
\begin{align*}
    \bc{ \alpha^2 - (1 - \alpha)^2 } = \frac{ (4 \gamma q + o(1)) (1 - \theta)}{c (1 - p - q)^2} \text{ s.t. } \alpha \in (0,1)
\end{align*}
in which $\alpha$ turns out to be
\begin{align*}
    \alpha = \frac{1 + o(1)}{2} \frac{c (1 -p -q)^2 + 4 \gamma q (1 - \theta)}{c (1 - p - q)^2}.
\end{align*}
Again, we plug the solution for $\alpha$ into \cref{eq_condition_n_sub1} and get
\begin{align*}
    \frac{dk}{n} = c > (4 \gamma q + o(1)) \frac{ \bc{1 + \sqrt{\theta}}^2 }{ \bc{1-p-q}^2 }
\end{align*}
Therefore, our algorithm reconstructs $\SIGMA$ correctly \whp if
\[ m \geq (4 \gamma + \eps)  \frac{ q \bc{1 + \sqrt{\theta}}^2 }{ \bc{1-p-q}^2 } n \log(n)\]
queries are conducted as claimed in \cref{thm_pq}.

\medskip

\paragraph{Linear Case}
Recall $k = \zeta n$, $\Gamma = \frac{n}{2}$, $\Delta = \frac{m}{2}$, $\gamma = (1 - \exp(-1/2))$ and $\Delta^\star = \gamma m$. Finally, $m = d k \log(n) = d \zeta n \log(n)$. Therefore, the conditions in \cref{eq_condition_k,eq_condition_n} read
\begin{align} \label{eq_condition_k_linear}
    \frac{(1-\alpha)^2 \bc{1-p-q}^2 d \zeta}{ (4 \gamma + o(1))\bc{ q + \zeta (1 - p - q) } } - 1 > 0
\end{align}
and
\begin{align} \label{eq_condition_n_linear}
    \frac{\alpha^2 \bc{1-p-q}^2 d \zeta}{ (4 \gamma + o(1)) \bc{ q + \zeta (1 - p - q) } } - 1 > 0.
\end{align}
By the same argument as before, the optimal $\alpha$ is given if the l.h.s.\ of \eqref{eq_condition_k_linear} and \eqref{eq_condition_n_linear} coincide, which is the case if
\begin{align} \label{eq_condition_alpha_linear}
    \alpha^2 = (1 - \alpha)^2.
\end{align}
Therefore, the optimal $\alpha$ is $\alpha = 1/2$. We conclude by \eqref{eq_condition_n_linear} that setting
\begin{align*}
    d > 16 \gamma \frac{q + (1 - p - q) \zeta}{(1 - p - q)^2 \zeta}
\end{align*}
suffices to satisfy the requirements.
Thus, our algorithm reconstructs $\SIGMA$ correctly \whp if
\[ m \geq (16 \gamma + \eps) \frac{q + (1 - p - q) \zeta}{(1 - p - q)^2 } n \log(n)\]
queries are conducted as claimed in \cref{thm_pq}.

\subsection{Noisy Query Model}
In this section we prove \cref{thm_noisy} for the noisy query model.
In this model, the hidden bits are measured correctly per se. However, the outcome of one query, the sum of the hidden bits, is exhibited to Gaussian fluctuations. More precisely, for each query $a$ we have
\begin{align*}
    \hat \SIGMA_a = \sum_{x \in \partial a} \SIGMA_x + \cW_a,
\end{align*}
where $\cW_1, \ldots, \cW_m \sim \cN(0, \lambda^2)$ are independent Gaussians with mean $0$ and variance $\lambda^2$.
The proof follows along the same lines as in the \pqchannel. Recall that $\Xi_j = \Xi_j^{00}$ is the number of agents with bit one in the (distinct) second neighborhood of agent $x_j$. 
Then \cref{lem_neighborhoodsum_pq} directly implies the following corollary if $p = q = 0$.
\begin{corollary}
\label{cor_neighborhoodsum_noise}
Let $\PSI_j = \Psi_j(\hat \SIGMA)$ be the neighborhood sum of agent $x_j$ with respect to the query results vector induced by the ground-truth $\SIGMA$. Then, under the noisy query model, 
\begin{align*}
    \PSI_j = \Psi_j(\hat \SIGMA) = \Xi_j +  \vecone \cbc{ \SIGMA_j = 1 } \Delta_j + \sum_{i = 1}^m \vecone \cbc{x_j \in \partial a_i} \cW_a.
\end{align*}
\end{corollary}
For the following proofs, we define $\cX_j$, the Gaussian noise in the neighborhood sum of agent $x_j$, as
\begin{align*}
    \cX_j = \sum_{i = 1}^m \vecone \cbc{ x_j \in \partial a_i} \cW_a \quad \text{ with} \quad 
    \cX_j \sim \cN \bc{ 0, \lambda^2 \Delta^\star_j }.
\end{align*}
We prove the two parts of \cref{thm_noisy} individually.
\medskip

\paragraph{Algorithmic Achievability}
By \cref{thm_pq}, 
\[ m \geq (4 \gamma + \eps) \bc{1 + \sqrt{\theta}}^2  k \log(n)\]
queries suffice in the sublinear case and, respectively,
\[  m \geq (16 \gamma + \eps) \frac{q + (1 - p - q) \zeta}{(1 - p - q)^2 } n \log(n) \]
queries suffice in the linear case to distinguish the neighborhood sum of all agents with hidden bits 1 and 0 by the difference of $\Delta_j$ due to the hidden bit. Therefore, if the noise $\cX_j$ in the neighborhood sum is of order $o \bc{\Delta_j}$ \whp for all agents, it is negligible and the same bounds hold.

By the previous discussion, the standard deviation of $\cX_j$ is given by 
\[ \nu_j = \lambda \sqrt{\Delta^\star_j}. \]
Let $\tau_n = \tau_n(\lambda) = o(1)$ be arbitrarily slowly vanishing. Then the tail bounds for the Gaussian distribution (\cref{lem_chernoff_gaussian} in the appendix) show that
{\dense
\begin{equation} \label{eq_tailbound_1}
\begin{aligned}
    \Pr  \bc{ \abs{\cX} \geq \Delta_j \tau_n \mid \cG, \cR} &= \Pr \bc{ \abs{\cX} \geq \frac{\tau_n \sqrt{\Delta_j}}{ \lambda \sqrt{2 \gamma} } \nu_j \mid \cG, \cR } \\
    & = \Theta\! \bc{ \exp \bc{ - (1 + o(1)) \frac{\tau_n^2 \Delta}{4 \gamma \lambda^2 } } } .
    \end{aligned}
\end{equation}}

\noindent Therefore, $\smash{\cX_j = o \bc{ \Delta_j }}$ \whp if 
$\smash{ \lambda = o \bc{ {\sqrt{\Delta}}/{\sqrt{\log n}} } }$, since in this case we can choose $\tau_n = o(1)$ such that the r.h.s.\ of \eqref{eq_tailbound_1} becomes $o \bc{n^{-1}}$. This implies the first part of \cref{thm_noisy} as $m = \Theta(\Delta_j)$ \whp by a union bound over all $n$ agents.

\medskip

\paragraph{Algorithmic Failure}
With the same argument, it is possible to pin down the magnitude of noise from which on \cref{algorithm} fails with positive probability for any number of queries conducted. We need to establish that for $\lambda^2 = \Omega  \bc{ \Delta_j } $ the noise dominates the difference in the neighborhood sum that can be observed between agents of different state.

Again, the Gaussian tail bounds show
\begin{equation}
\begin{aligned} \label{eq_tailbound_2}
    \Pr  \bc{ \cX_j \geq \Delta_j \mid \cG, \cR} &= \Pr \bc{ \cX \geq \frac{\sqrt{\Delta_j}}{ \lambda \sqrt{2 \gamma} } \nu_j \mid \cG, \cR} \\
     & = \Theta \bc{ \exp \bc{ - \frac{\Delta}{4 \gamma \lambda^2 } } } = \Theta(1).
\end{aligned}
\end{equation}

\noindent Let $A = \abs{ \cbc{ j: \cX_j \geq \Delta_j } }$, then $\Erw \brk{A \mid \cG, \cR} = \Theta(n)$ and by the reverse Markov inequality, we have for $0 < t < \Erw \brk{A \mid \cG, \cR}/2$ that
\begin{align*}
 \Pr & \bc{ A > t \mid \cG, \cR} \geq \frac{\Erw \brk{A \mid \cG, \cR} - t}{n - t} = \Theta(1).  
\end{align*}
This implies the second part of \cref{thm_noisy}.

\section{Simulations}\label{sec:simulations}

\begin{figure*}[p]
\begin{minipage}{3.5in}
\input{figures/plot1-pq}
\caption{The plots show the required number of queries for the Z-channel ($q = 0)$ with $\theta = 0.25$. The dashed line shows our theoretical bound for $p = 0.1$.}
\label{fig:required-number-of-queries-noisy-channel}
\end{minipage}\hfill
\begin{minipage}{3.5in}
\input{figures/plot1-sigma}
\caption{The plots show a comparison of the required number of queries in the noisy query model and the required number of queries without noise for $\theta = 0.25$.}
\label{fig:required-number-of-queries-noisy-pool}
\end{minipage}

\vspace{1cm}

\begin{minipage}{3.5in}
\input{figures/plot1-pq-sym}
\caption{The plot shows the required number of queries for the noisy channel model with $p = q = 10^{-1}, 10^{-2}, 10^{-3}, 10^{-4}, 10^{-5}$ for $\theta = 0.25$.}
\label{fig:required-number-of-queries-noisy-channel-symmetric}
\end{minipage}\hfill
\begin{minipage}{3.5in}
\input{figures/plot-boxplot}
\caption{The figures shows boxplots for the data presented in \cref{fig:required-number-of-queries-noisy-channel,fig:required-number-of-queries-noisy-pool} for $n = 10^3, 10^4, 10^5$ (and additional noise levels for  \cref{fig:required-number-of-queries-noisy-pool}).}
\label{fig:required-number-of-queries-boxplot}
\end{minipage}

\vspace{1cm}

\begin{minipage}{3.5in}
\input{figures/plot-amp}
\caption{The plot shows the success rate for $n = 1000$ and $p = 0.1, 0.3, 0.5$ in the Z-channel model. The dashed line indicates the theoretical bound for $p = 0.1$.}
\label{fig:success-probability}
\end{minipage}
\hfill
\begin{minipage}{3.5in}
\input{figures/plot-overlap}
\caption{The plot shows the \emph{overlap} (the average fraction of correctly identified one bits) for $n = 1000$ and $p = 0.1, 0.3, 0.5$ in the Z-channel model. The dashed line indicates the theoretical bound for $p = 0.1$.}
\label{fig:overlap}
\end{minipage}

\end{figure*}

In this section we present simulation results for our algorithm.
Our simulation software is 
implemented in the \CC{} programming language.
It first generates the random pooling scheme by constructing a random bipartite graph.
Then it simulates the interactions between agents and query nodes by computing the agents' scores as defined in \cref{algorithm}.
As a source of randomness we use the Mersenne Twister \texttt{mt19937\_64} provided by the \CC{11} \texttt{\textless random\textgreater{}} library.
Our simulations have been carried out on a machine with two Intel(R) Xeon(R) E5-2630 v4 CPUs with 128 GiB of memory running the Linux 5.13 kernel.
The simulation software and all related tools will be made publicly available upon publication of this paper.

\smallskip

\paragraph{Required Number of Queries}
We start our empirical analysis with the number of queries that are required to reconstruct the data in both noise models, the noisy channel model and the noisy query model.
We present our data for the fixed value of $\theta = 0.25$, however, extensive simulations comprising multiple different values of $\theta$ have shown consistent results.

In \cref{fig:required-number-of-queries-noisy-channel} we consider the Z-channel model where we assume that the $1$ bits flip with probability $p$ when read by a query node.
We plot the \emph{required number of queries} (see below for implementation details) for varying values of $n$ on the $x$-axis, three different noise levels $p = 0.1, 0.2, 0.3$ and $\theta = 0.25$.
The dashed line indicates our theoretical bounds from \cref{thm_pq} for $p = 0.1$ and $\epsilon = 0.05$.
We remark that for $n \in [10^2, 10^5]$ our simulation results align well with our theoretical findings for small and moderate error probabilities up to $p = 0.1$.
Unfortunately this is not entirely the case for larger (but arguably less realistic) values of $p$, which we also present in \cref{fig:required-number-of-queries-noisy-channel}.
In our asymptotic analysis we use that $\log^2(n) \sqrt{\Delta \cdot (1-p) }$ is much smaller than $\Delta (1-p)$ (see \cref{eq_expected_difference} in the proof of \cref{lem_neighborhoodsum_pq}).\footnote{We emphasize that $2 \sqrt{\Delta \cdot (1-p) \log(k)} \ll \Delta(1-p)$ would also suffice analytically if the Chernoff bound was applied in its strongest variant.} 
While this is always the case for sufficiently large values of $n$, the required number of agents to fulfill this property is well beyond any practical input sizes.

In \cref{fig:required-number-of-queries-noisy-pool} we consider the noisy query model where we assume that each query returns a random variate sampled from a normal distribution.
We compare the required number of queries under noise ($\lambda = 2$) to the required number of queries without noise.
\Cref{fig:required-number-of-queries-boxplot} shows additional data for both the Z-channel model (for $p = 0.1, 0.3, 0.5$) and the noisy query model (for $\lambda = 0, 1, 2, 3$) in the form of box plots for $n = 10^3$, $n = 10^4$, and $n = 10^5$.

Finally, in \cref{fig:required-number-of-queries-noisy-channel-symmetric} we consider the general noisy channel with $p = q$.
We plot the required number of queries for symmetric error rates $p = q = 10^{-1}, \dots, 10^{-5}$.
Note that in this plot one can observe the transition between two regimes as predicted by our theoretical results (see remark after \cref{thm_pq}).
Indeed, consider the case $p = q = 10^{-3}$ shown in blue.
Up to roughly $n = 3000$, the expression $k/n$ dominates, while for larger values of $n$ the error probability $q$ dominates, resulting in a noticeably steeper ascend of the required number of queries.
The dashed lines show our theoretical bounds for this latter setting.

\smallskip

\paragraph{Implementation Details}
In our simulation, we compute the required number of queries reported in \cref{fig:required-number-of-queries-noisy-channel,fig:required-number-of-queries-noisy-pool,fig:required-number-of-queries-noisy-channel-symmetric,fig:required-number-of-queries-boxplot} as follows.
First we initialize the ground truth according to $n$ and $\theta$.
Then we simulate one query node after the other in a sequential manner.
Each query node samples a number of agents uniformly at random with replacement.
In the \pqchannel model, the simulation software applies the random bit flips when computing the sum.
In the noisy query model, we first compute the exact sum $\mu$.
Then we draw a random variate according to a normal distribution $N(\mu, \lambda)$ where $\lambda$ is the noise parameter of the model.
Finally, we updates the values of $\Delta^*$ and $\Psi^*$ accordingly. Note that these steps perform a faithful simulation of the distributed system.

Our simulation uses a permutation on $[n]$ to store an order among the agents. After adding each query node, this permutation is updated to reflect the new values of $\Delta^*$ and $\Psi^*$.
Our simulation terminates once the ground truth can be reconstructed exactly; this involves a check whether all agents have been correctly identified, and whether there is a clear separation between the scores of the $0$ agents and the $1$ agents.

\smallskip

\paragraph{Success Probabilities and Comparison with AMP}
In \cref{fig:success-probability} we analyze the success probabilities of the distributed reconstruction with our greedy algorithm, and we compare it with AMP.
We consider a fixed numbers of $n = 1000$ agents.
For each value of $m$ we have conducted $100$ independent simulation runs. In the plot we show the relative number of successful reconstructions of all agents' hidden bits.
The dashed line shows our theoretical bounds from \cref{thm_pq} for $p = 0.1$ and $\epsilon = 0.1$.

In \cref{fig:overlap} we analyze the \emph{overlap}. It is defined as the fraction of agents with bit 1 that are correctly identified.
Again, $100$ independent simulation runs have been conducted for each value of $m$.
For moderate error probabilities our data indicate a substantial overlap, even if the exact reconstruction is still quite unlikely.
Indeed, consider the threshold from our theoretical results.
From \cref{fig:success-probability} we observe a success rate of perfect reconstructions of about 40\%, while from \cref{fig:overlap} we obtain that on average almost 90\% of the one-agents have been correctly identified.
We remark that this property hints at the practical applicability of our reconstruction algorithm, where often a small probability of misclassification is acceptable.

\section{Conclusion} \label{sec:discussion}
It is conjectured that AMP might be optimal in reconstruction problems like the pooled data problem studied in this contribution \cite{alaoui_2017, donoho_l1, donoho_amp, donoho_amp2}. While the algorithm has a distributed touch as it can be implemented such that single network nodes (agents or queries) exchange only messages with their neighbors, the algorithm requires an information flow through the whole communication network within multiple rounds. In contrast, our greedy approach is completely distributed and requires only one information exchange per network node. It is therefore not surprising that the centralized AMP algorithm outperforms \cref{algorithm} in certain settings. Nevertheless, our simulation data suggests that both algorithms exhibit a similar phase transition between failing most of the time and succeeding most of the time. The width of the phase transition window, however, is much smaller for AMP. This might be due to the following observation.

AMP is initialized with $\sigma^{(0)} = 0$ and computes in the first round $\sigma^{(1)} = \eta_1 \bc{\vec A^T \hat \SIGMA }$. The term $\vec A^T \hat \SIGMA$ corresponds to our \emph{neighborhood sum} (with the difference that multi-edges are counted multiple times).

It follows that the information that AMP can use after exactly one update step is the same as in \cref{algorithm}. If there is enough information in the neighborhood sum such that a good portion of hidden bits can be estimated correctly, the iterative procedure of AMP seems to allow the recovery of the remaining few mistakes. This conjecture is supported by the fact that \cref{algorithm} outputs an estimate with a high overlap with the ground-truth $\SIGMA$ with the same number of queries with which AMP succeeds already quite often. If, on the other hand, there is not enough information in the neighborhood sum, \cref{algorithm} produces an estimate with small overlap with $\SIGMA$. In this case, AMP seems to fail as well. This might be because the first estimate $\sigma^{(1)}$ guides AMP into the attraction basin of a false fixed-point.

Overall, we analyzed a simple, distributed and fast reconstruction algorithm for the pooled data problem that exhibits similarities with the first step of AMP. It turns out that our distributed reconstruction algorithm is robust against noise and stays reliable. We give exact and rigorous asymptotic achievability bounds that match our simulation data. An intriguing open question is whether a two-step algorithm that locally tries to correct errors can be analyzed rigorously and performs even better.

\vfill\null


\bstctlcite{IEEEexample:BSTcontrol}
\bibliography{bibliography}

\appendix

\begin{theorem}[Chernoff Bound for Negatively Associated Random Variables \cite{DBLP:books/daglib/0025902,Chen2019}] \label{lem_chernoff}
Let $(X_1, \ldots, X_n)$ be a sequence of negatively associated Bernoulli random variables such that $X_i \sim \Be(p_i)$. Let $X = \sum_{i=1}^n X_i$ and $p = n^{-1} \sum-{i=1}^n p_i$. Then for any $\eps >0$, 
\begin{align*}
    & \Pr\bc{ X \geq (1 + \eps) \Erw \brk{X}} \leq \exp\bc{-\frac{\eps^2}{2 + \eps} \Erw\brk{X}} \qquad \text{and}  \\
    & \Pr\bc{X \leq (1 - \eps) \Erw\brk{X}} \leq \exp\bc{-\frac{\eps^2}{2} \Erw\brk{X}}.
\end{align*}
\end{theorem}
A similar bound holds for Gaussian random variables. Furthermore, Mill's ratio yields a fitting lower bound (\cite{duembgen2010bounding} or \cite[Section 7.1]{feller1}). 
\begin{theorem}[Tail bound for Gaussian random variables] \label{lem_chernoff_gaussian}
Let $X \sim \cN(0, \lambda^2)$ be a Gaussian random variable with mean zero and variance $\lambda^2$. Then, for any $y > 0$, we have
\begin{align*}
    \Pr\bc{ X \geq y } \leq \frac{\lambda}{y} \frac{1}{\sqrt{2 \pi}} \exp\bc{- \frac{y^2}{2 \lambda^2}  }, \\
    \Pr\bc{ X \geq y } \geq  \bc{\frac{\lambda}{y} - \frac{\lambda^3}{y^3}} \frac{1}{\sqrt{2 \pi}} \exp\bc{- \frac{y^2}{2 \lambda^2}  }.
\end{align*}
Analogously,
\begin{align*}
    \Pr\bc{ X \leq -y } \leq \frac{\lambda}{y} \frac{1}{\sqrt{2 \pi}} \exp\bc{- \frac{y^2}{2 \lambda^2}  }, \\
    \Pr\bc{ X \leq -y } \geq  \bc{\frac{\lambda}{y} - \frac{\lambda^3}{y^3}} \frac{1}{\sqrt{2 \pi}} \exp\bc{- \frac{y^2}{2 \lambda^2}  }.
\end{align*}
\end{theorem}

\end{document}

%% file: macros.tex
\renewcommand\log{\ln}

\renewcommand{\epsilon}{\eps}

\newcommand\PSI{\vec\psi}

\renewcommand{\vec}[1]{\boldsymbol{#1}}

\newcommand\SIGMA{\vec\sigma}

\newcommand\G{\mathcal{G}}

\newcommand\cN{\mathcal{N}}

\newcommand\cP{\mathcal{P}}
\newcommand\cX{\mathcal{X}}

\newcommand\cW{\mathcal{W}}

\def\cR{{\mathcal R}}

\newcommand\eps{\varepsilon}

\newcommand\NN{\mathbb{N}}

\newcommand\Erw{\mathbb{E}\!}
\newcommand{\vecone}{\vec{1}}

\newcommand{\set}[1]{\left\{#1\right\}}

\newcommand{\Bin}{{\rm Bin}}
\newcommand{\Mult}{{\rm Mult}}
\newcommand{\Be}{{\rm Be}}

\newcommand\bc[1]{\left({#1}\right)}
\newcommand\cbc[1]{\left\{{#1}\right\}}

\newcommand\brk[1]{\left\lbrack{#1}\right\rbrack}

\newcommand\abs[1]{\left|{#1}\right|}

\newcommand\RR{\mathbb{R}}

\newcommand{\Erdos}{Erd\H{o}s}
\newcommand{\Renyi}{R\'enyi}

\newcommand\pr{\mathbb{P}} 
\renewcommand\Pr{\pr\!}

 \def\G{{\vec G}}

\def\pr{{\mathbb P}}

\newcommand{\remove}[1]{}

\newcommand{\cG}{\mathfrak{G}}

%% file: figures/plot1-pq.tex
\begingroup
  \makeatletter
  \providecommand\color[2][]{%
    \GenericError{(gnuplot) \space\space\space\@spaces}{%
      Package color not loaded in conjunction with
      terminal option `colourtext'%
    }{See the gnuplot documentation for explanation.%
    }{Either use 'blacktext' in gnuplot or load the package
      color.sty in LaTeX.}%
    \renewcommand\color[2][]{}%
  }%
  \providecommand\includegraphics[2][]{%
    \GenericError{(gnuplot) \space\space\space\@spaces}{%
      Package graphicx or graphics not loaded%
    }{See the gnuplot documentation for explanation.%
    }{The gnuplot epslatex terminal needs graphicx.sty or graphics.sty.}%
    \renewcommand\includegraphics[2][]{}%
  }%
  \providecommand\rotatebox[2]{#2}%
  \@ifundefined{ifGPcolor}{%
    \newif\ifGPcolor
    \GPcolorfalse
  }{}%
  \@ifundefined{ifGPblacktext}{%
    \newif\ifGPblacktext
    \GPblacktexttrue
  }{}%
  \let\gplgaddtomacro\g@addto@macro
  \gdef\gplbacktext{}%
  \gdef\gplfronttext{}%
  \makeatother
  \ifGPblacktext
    \def\colorrgb#1{}%
    \def\colorgray#1{}%
  \else
    \ifGPcolor
      \def\colorrgb#1{\color[rgb]{#1}}%
      \def\colorgray#1{\color[gray]{#1}}%
      \expandafter\def\csname LTw\endcsname{\color{white}}%
      \expandafter\def\csname LTb\endcsname{\color{black}}%
      \expandafter\def\csname LTa\endcsname{\color{black}}%
      \expandafter\def\csname LT0\endcsname{\color[rgb]{1,0,0}}%
      \expandafter\def\csname LT1\endcsname{\color[rgb]{0,1,0}}%
      \expandafter\def\csname LT2\endcsname{\color[rgb]{0,0,1}}%
      \expandafter\def\csname LT3\endcsname{\color[rgb]{1,0,1}}%
      \expandafter\def\csname LT4\endcsname{\color[rgb]{0,1,1}}%
      \expandafter\def\csname LT5\endcsname{\color[rgb]{1,1,0}}%
      \expandafter\def\csname LT6\endcsname{\color[rgb]{0,0,0}}%
      \expandafter\def\csname LT7\endcsname{\color[rgb]{1,0.3,0}}%
      \expandafter\def\csname LT8\endcsname{\color[rgb]{0.5,0.5,0.5}}%
    \else
      \def\colorrgb#1{\color{black}}%
      \def\colorgray#1{\color[gray]{#1}}%
      \expandafter\def\csname LTw\endcsname{\color{white}}%
      \expandafter\def\csname LTb\endcsname{\color{black}}%
      \expandafter\def\csname LTa\endcsname{\color{black}}%
      \expandafter\def\csname LT0\endcsname{\color{black}}%
      \expandafter\def\csname LT1\endcsname{\color{black}}%
      \expandafter\def\csname LT2\endcsname{\color{black}}%
      \expandafter\def\csname LT3\endcsname{\color{black}}%
      \expandafter\def\csname LT4\endcsname{\color{black}}%
      \expandafter\def\csname LT5\endcsname{\color{black}}%
      \expandafter\def\csname LT6\endcsname{\color{black}}%
      \expandafter\def\csname LT7\endcsname{\color{black}}%
      \expandafter\def\csname LT8\endcsname{\color{black}}%
    \fi
  \fi
    \setlength{\unitlength}{0.0500bp}%
    \ifx\gptboxheight\undefined%
      \newlength{\gptboxheight}%
      \newlength{\gptboxwidth}%
      \newsavebox{\gptboxtext}%
    \fi%
    \setlength{\fboxrule}{0.5pt}%
    \setlength{\fboxsep}{1pt}%
    \definecolor{tbcol}{rgb}{1,1,1}%
\begin{picture}(5040.00,3358.00)%
    \gplgaddtomacro\gplbacktext{%
      \csname LTb\endcsname
      \put(624,503){\makebox(0,0)[r]{\strut{}$10^{1}$}}%
      \put(624,1398){\makebox(0,0)[r]{\strut{}$10^{2}$}}%
      \put(624,2294){\makebox(0,0)[r]{\strut{}$10^{3}$}}%
      \put(624,3189){\makebox(0,0)[r]{\strut{}$10^{4}$}}%
      \put(756,283){\makebox(0,0){\strut{}$10^{2}$}}%
      \put(2100,283){\makebox(0,0){\strut{}$10^{3}$}}%
      \put(3443,283){\makebox(0,0){\strut{}$10^{4}$}}%
      \put(4787,283){\makebox(0,0){\strut{}$10^{5}$}}%
    }%
    \gplgaddtomacro\gplfronttext{%
      \csname LTb\endcsname
      \put(151,1846){\rotatebox{-270}{\makebox(0,0){\strut{}required number of queries $m$}}}%
      \put(2771,63){\makebox(0,0){\strut{}number of agents $n$}}%
      \csname LTb\endcsname
      \put(4328,1116){\makebox(0,0)[r]{\strut{}$p = 0.1$}}%
      \csname LTb\endcsname
      \put(4328,896){\makebox(0,0)[r]{\strut{}$p = 0.3$}}%
      \csname LTb\endcsname
      \put(4328,676){\makebox(0,0)[r]{\strut{}$p = 0.5$}}%
      \csname LTb\endcsname
      \put(2771,3189){\makebox(0,0){\strut{}}}%
    }%
    \gplbacktext
    \put(0,0){\includegraphics[width={252.00bp},height={167.90bp}]{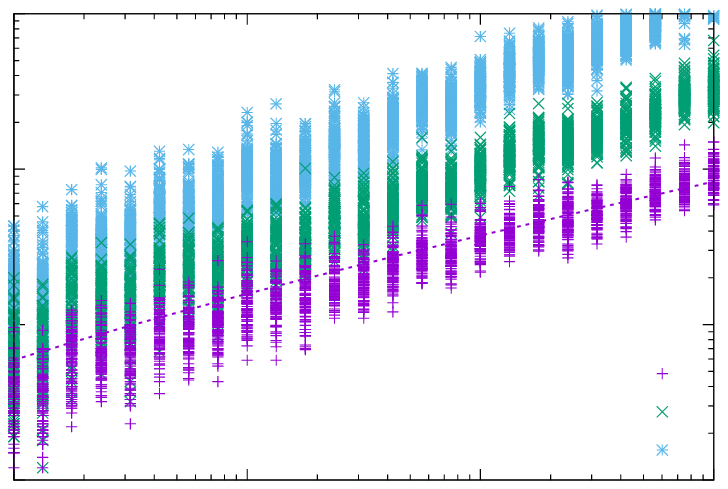}}%
    \gplfronttext
  \end{picture}%
\endgroup

%% file: figures/plot1-sigma.tex
\begingroup
  \makeatletter
  \providecommand\color[2][]{%
    \GenericError{(gnuplot) \space\space\space\@spaces}{%
      Package color not loaded in conjunction with
      terminal option `colourtext'%
    }{See the gnuplot documentation for explanation.%
    }{Either use 'blacktext' in gnuplot or load the package
      color.sty in LaTeX.}%
    \renewcommand\color[2][]{}%
  }%
  \providecommand\includegraphics[2][]{%
    \GenericError{(gnuplot) \space\space\space\@spaces}{%
      Package graphicx or graphics not loaded%
    }{See the gnuplot documentation for explanation.%
    }{The gnuplot epslatex terminal needs graphicx.sty or graphics.sty.}%
    \renewcommand\includegraphics[2][]{}%
  }%
  \providecommand\rotatebox[2]{#2}%
  \@ifundefined{ifGPcolor}{%
    \newif\ifGPcolor
    \GPcolorfalse
  }{}%
  \@ifundefined{ifGPblacktext}{%
    \newif\ifGPblacktext
    \GPblacktexttrue
  }{}%
  \let\gplgaddtomacro\g@addto@macro
  \gdef\gplbacktext{}%
  \gdef\gplfronttext{}%
  \makeatother
  \ifGPblacktext
    \def\colorrgb#1{}%
    \def\colorgray#1{}%
  \else
    \ifGPcolor
      \def\colorrgb#1{\color[rgb]{#1}}%
      \def\colorgray#1{\color[gray]{#1}}%
      \expandafter\def\csname LTw\endcsname{\color{white}}%
      \expandafter\def\csname LTb\endcsname{\color{black}}%
      \expandafter\def\csname LTa\endcsname{\color{black}}%
      \expandafter\def\csname LT0\endcsname{\color[rgb]{1,0,0}}%
      \expandafter\def\csname LT1\endcsname{\color[rgb]{0,1,0}}%
      \expandafter\def\csname LT2\endcsname{\color[rgb]{0,0,1}}%
      \expandafter\def\csname LT3\endcsname{\color[rgb]{1,0,1}}%
      \expandafter\def\csname LT4\endcsname{\color[rgb]{0,1,1}}%
      \expandafter\def\csname LT5\endcsname{\color[rgb]{1,1,0}}%
      \expandafter\def\csname LT6\endcsname{\color[rgb]{0,0,0}}%
      \expandafter\def\csname LT7\endcsname{\color[rgb]{1,0.3,0}}%
      \expandafter\def\csname LT8\endcsname{\color[rgb]{0.5,0.5,0.5}}%
    \else
      \def\colorrgb#1{\color{black}}%
      \def\colorgray#1{\color[gray]{#1}}%
      \expandafter\def\csname LTw\endcsname{\color{white}}%
      \expandafter\def\csname LTb\endcsname{\color{black}}%
      \expandafter\def\csname LTa\endcsname{\color{black}}%
      \expandafter\def\csname LT0\endcsname{\color{black}}%
      \expandafter\def\csname LT1\endcsname{\color{black}}%
      \expandafter\def\csname LT2\endcsname{\color{black}}%
      \expandafter\def\csname LT3\endcsname{\color{black}}%
      \expandafter\def\csname LT4\endcsname{\color{black}}%
      \expandafter\def\csname LT5\endcsname{\color{black}}%
      \expandafter\def\csname LT6\endcsname{\color{black}}%
      \expandafter\def\csname LT7\endcsname{\color{black}}%
      \expandafter\def\csname LT8\endcsname{\color{black}}%
    \fi
  \fi
    \setlength{\unitlength}{0.0500bp}%
    \ifx\gptboxheight\undefined%
      \newlength{\gptboxheight}%
      \newlength{\gptboxwidth}%
      \newsavebox{\gptboxtext}%
    \fi%
    \setlength{\fboxrule}{0.5pt}%
    \setlength{\fboxsep}{1pt}%
    \definecolor{tbcol}{rgb}{1,1,1}%
\begin{picture}(5040.00,3358.00)%
    \gplgaddtomacro\gplbacktext{%
      \csname LTb\endcsname
      \put(624,503){\makebox(0,0)[r]{\strut{}$10^{1}$}}%
      \put(624,1398){\makebox(0,0)[r]{\strut{}$10^{2}$}}%
      \put(624,2294){\makebox(0,0)[r]{\strut{}$10^{3}$}}%
      \put(624,3189){\makebox(0,0)[r]{\strut{}$10^{4}$}}%
      \put(756,283){\makebox(0,0){\strut{}$10^{2}$}}%
      \put(2100,283){\makebox(0,0){\strut{}$10^{3}$}}%
      \put(3443,283){\makebox(0,0){\strut{}$10^{4}$}}%
      \put(4787,283){\makebox(0,0){\strut{}$10^{5}$}}%
    }%
    \gplgaddtomacro\gplfronttext{%
      \csname LTb\endcsname
      \put(151,1846){\rotatebox{-270}{\makebox(0,0){\strut{}required number of queries $m$}}}%
      \put(2771,63){\makebox(0,0){\strut{}number of agents $n$}}%
      \csname LTb\endcsname
      \put(4328,3016){\makebox(0,0)[r]{\strut{}without noise}}%
      \csname LTb\endcsname
      \put(4328,2796){\makebox(0,0)[r]{\strut{}with noise ($\lambda = 1$)}}%
      \csname LTb\endcsname
      \put(2771,3189){\makebox(0,0){\strut{}}}%
    }%
    \gplbacktext
    \put(0,0){\includegraphics[width={252.00bp},height={167.90bp}]{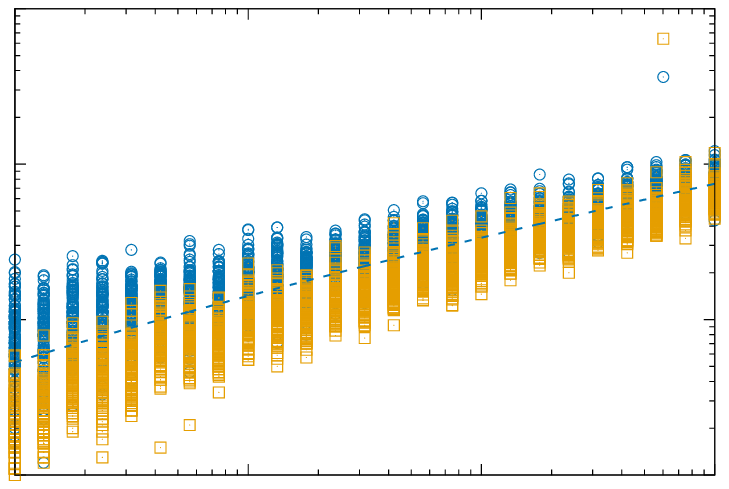}}%
    \gplfronttext
  \end{picture}%
\endgroup

%% file: figures/plot1-pq-sym.tex
\begingroup
  \makeatletter
  \providecommand\color[2][]{%
    \GenericError{(gnuplot) \space\space\space\@spaces}{%
      Package color not loaded in conjunction with
      terminal option `colourtext'%
    }{See the gnuplot documentation for explanation.%
    }{Either use 'blacktext' in gnuplot or load the package
      color.sty in LaTeX.}%
    \renewcommand\color[2][]{}%
  }%
  \providecommand\includegraphics[2][]{%
    \GenericError{(gnuplot) \space\space\space\@spaces}{%
      Package graphicx or graphics not loaded%
    }{See the gnuplot documentation for explanation.%
    }{The gnuplot epslatex terminal needs graphicx.sty or graphics.sty.}%
    \renewcommand\includegraphics[2][]{}%
  }%
  \providecommand\rotatebox[2]{#2}%
  \@ifundefined{ifGPcolor}{%
    \newif\ifGPcolor
    \GPcolorfalse
  }{}%
  \@ifundefined{ifGPblacktext}{%
    \newif\ifGPblacktext
    \GPblacktexttrue
  }{}%
  \let\gplgaddtomacro\g@addto@macro
  \gdef\gplbacktext{}%
  \gdef\gplfronttext{}%
  \makeatother
  \ifGPblacktext
    \def\colorrgb#1{}%
    \def\colorgray#1{}%
  \else
    \ifGPcolor
      \def\colorrgb#1{\color[rgb]{#1}}%
      \def\colorgray#1{\color[gray]{#1}}%
      \expandafter\def\csname LTw\endcsname{\color{white}}%
      \expandafter\def\csname LTb\endcsname{\color{black}}%
      \expandafter\def\csname LTa\endcsname{\color{black}}%
      \expandafter\def\csname LT0\endcsname{\color[rgb]{1,0,0}}%
      \expandafter\def\csname LT1\endcsname{\color[rgb]{0,1,0}}%
      \expandafter\def\csname LT2\endcsname{\color[rgb]{0,0,1}}%
      \expandafter\def\csname LT3\endcsname{\color[rgb]{1,0,1}}%
      \expandafter\def\csname LT4\endcsname{\color[rgb]{0,1,1}}%
      \expandafter\def\csname LT5\endcsname{\color[rgb]{1,1,0}}%
      \expandafter\def\csname LT6\endcsname{\color[rgb]{0,0,0}}%
      \expandafter\def\csname LT7\endcsname{\color[rgb]{1,0.3,0}}%
      \expandafter\def\csname LT8\endcsname{\color[rgb]{0.5,0.5,0.5}}%
    \else
      \def\colorrgb#1{\color{black}}%
      \def\colorgray#1{\color[gray]{#1}}%
      \expandafter\def\csname LTw\endcsname{\color{white}}%
      \expandafter\def\csname LTb\endcsname{\color{black}}%
      \expandafter\def\csname LTa\endcsname{\color{black}}%
      \expandafter\def\csname LT0\endcsname{\color{black}}%
      \expandafter\def\csname LT1\endcsname{\color{black}}%
      \expandafter\def\csname LT2\endcsname{\color{black}}%
      \expandafter\def\csname LT3\endcsname{\color{black}}%
      \expandafter\def\csname LT4\endcsname{\color{black}}%
      \expandafter\def\csname LT5\endcsname{\color{black}}%
      \expandafter\def\csname LT6\endcsname{\color{black}}%
      \expandafter\def\csname LT7\endcsname{\color{black}}%
      \expandafter\def\csname LT8\endcsname{\color{black}}%
    \fi
  \fi
    \setlength{\unitlength}{0.0500bp}%
    \ifx\gptboxheight\undefined%
      \newlength{\gptboxheight}%
      \newlength{\gptboxwidth}%
      \newsavebox{\gptboxtext}%
    \fi%
    \setlength{\fboxrule}{0.5pt}%
    \setlength{\fboxsep}{1pt}%
    \definecolor{tbcol}{rgb}{1,1,1}%
\begin{picture}(5040.00,3358.00)%
    \gplgaddtomacro\gplbacktext{%
      \csname LTb\endcsname
      \put(624,503){\makebox(0,0)[r]{\strut{}$10^{1}$}}%
      \put(624,1398){\makebox(0,0)[r]{\strut{}$10^{2}$}}%
      \put(624,2294){\makebox(0,0)[r]{\strut{}$10^{3}$}}%
      \put(624,3189){\makebox(0,0)[r]{\strut{}$10^{4}$}}%
      \put(756,283){\makebox(0,0){\strut{}$10^{2}$}}%
      \put(2100,283){\makebox(0,0){\strut{}$10^{3}$}}%
      \put(3443,283){\makebox(0,0){\strut{}$10^{4}$}}%
      \put(4787,283){\makebox(0,0){\strut{}$10^{5}$}}%
    }%
    \gplgaddtomacro\gplfronttext{%
      \csname LTb\endcsname
      \put(151,1846){\rotatebox{-270}{\makebox(0,0){\strut{}required number of queries $m$}}}%
      \put(2771,63){\makebox(0,0){\strut{}number of agents $n$}}%
      \csname LTb\endcsname
      \put(4328,1556){\makebox(0,0)[r]{\strut{}$q = 10^{-1}$}}%
      \csname LTb\endcsname
      \put(4328,1336){\makebox(0,0)[r]{\strut{}$q = 10^{-2}$}}%
      \csname LTb\endcsname
      \put(4328,1116){\makebox(0,0)[r]{\strut{}$q = 10^{-3}$}}%
      \csname LTb\endcsname
      \put(4328,896){\makebox(0,0)[r]{\strut{}$q = 10^{-4}$}}%
      \csname LTb\endcsname
      \put(4328,676){\makebox(0,0)[r]{\strut{}$q = 10^{-5}$}}%
      \csname LTb\endcsname
      \put(2771,3189){\makebox(0,0){\strut{}}}%
    }%
    \gplbacktext
    \put(0,0){\includegraphics[width={252.00bp},height={167.90bp}]{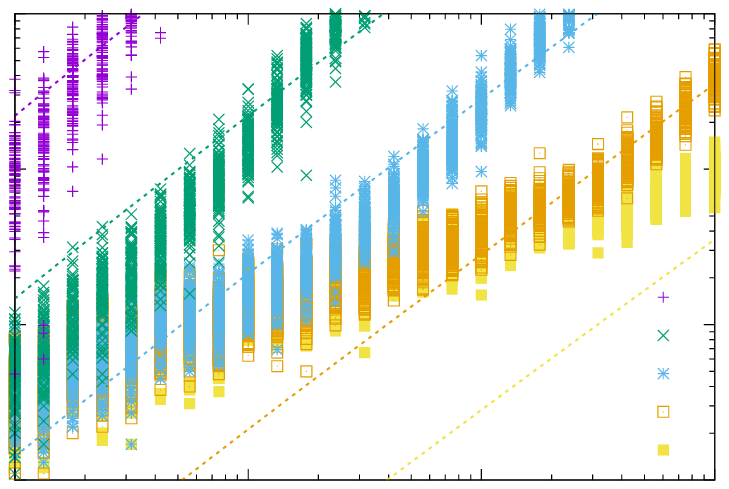}}%
    \gplfronttext
  \end{picture}%
\endgroup

%% file: figures/plot-boxplot.tex
\begingroup
  \makeatletter
  \providecommand\color[2][]{%
    \GenericError{(gnuplot) \space\space\space\@spaces}{%
      Package color not loaded in conjunction with
      terminal option `colourtext'%
    }{See the gnuplot documentation for explanation.%
    }{Either use 'blacktext' in gnuplot or load the package
      color.sty in LaTeX.}%
    \renewcommand\color[2][]{}%
  }%
  \providecommand\includegraphics[2][]{%
    \GenericError{(gnuplot) \space\space\space\@spaces}{%
      Package graphicx or graphics not loaded%
    }{See the gnuplot documentation for explanation.%
    }{The gnuplot epslatex terminal needs graphicx.sty or graphics.sty.}%
    \renewcommand\includegraphics[2][]{}%
  }%
  \providecommand\rotatebox[2]{#2}%
  \@ifundefined{ifGPcolor}{%
    \newif\ifGPcolor
    \GPcolorfalse
  }{}%
  \@ifundefined{ifGPblacktext}{%
    \newif\ifGPblacktext
    \GPblacktexttrue
  }{}%
  \let\gplgaddtomacro\g@addto@macro
  \gdef\gplbacktext{}%
  \gdef\gplfronttext{}%
  \makeatother
  \ifGPblacktext
    \def\colorrgb#1{}%
    \def\colorgray#1{}%
  \else
    \ifGPcolor
      \def\colorrgb#1{\color[rgb]{#1}}%
      \def\colorgray#1{\color[gray]{#1}}%
      \expandafter\def\csname LTw\endcsname{\color{white}}%
      \expandafter\def\csname LTb\endcsname{\color{black}}%
      \expandafter\def\csname LTa\endcsname{\color{black}}%
      \expandafter\def\csname LT0\endcsname{\color[rgb]{1,0,0}}%
      \expandafter\def\csname LT1\endcsname{\color[rgb]{0,1,0}}%
      \expandafter\def\csname LT2\endcsname{\color[rgb]{0,0,1}}%
      \expandafter\def\csname LT3\endcsname{\color[rgb]{1,0,1}}%
      \expandafter\def\csname LT4\endcsname{\color[rgb]{0,1,1}}%
      \expandafter\def\csname LT5\endcsname{\color[rgb]{1,1,0}}%
      \expandafter\def\csname LT6\endcsname{\color[rgb]{0,0,0}}%
      \expandafter\def\csname LT7\endcsname{\color[rgb]{1,0.3,0}}%
      \expandafter\def\csname LT8\endcsname{\color[rgb]{0.5,0.5,0.5}}%
    \else
      \def\colorrgb#1{\color{black}}%
      \def\colorgray#1{\color[gray]{#1}}%
      \expandafter\def\csname LTw\endcsname{\color{white}}%
      \expandafter\def\csname LTb\endcsname{\color{black}}%
      \expandafter\def\csname LTa\endcsname{\color{black}}%
      \expandafter\def\csname LT0\endcsname{\color{black}}%
      \expandafter\def\csname LT1\endcsname{\color{black}}%
      \expandafter\def\csname LT2\endcsname{\color{black}}%
      \expandafter\def\csname LT3\endcsname{\color{black}}%
      \expandafter\def\csname LT4\endcsname{\color{black}}%
      \expandafter\def\csname LT5\endcsname{\color{black}}%
      \expandafter\def\csname LT6\endcsname{\color{black}}%
      \expandafter\def\csname LT7\endcsname{\color{black}}%
      \expandafter\def\csname LT8\endcsname{\color{black}}%
    \fi
  \fi
    \setlength{\unitlength}{0.0500bp}%
    \ifx\gptboxheight\undefined%
      \newlength{\gptboxheight}%
      \newlength{\gptboxwidth}%
      \newsavebox{\gptboxtext}%
    \fi%
    \setlength{\fboxrule}{0.5pt}%
    \setlength{\fboxsep}{1pt}%
    \definecolor{tbcol}{rgb}{1,1,1}%
\begin{picture}(5040.00,3358.00)%
    \gplgaddtomacro\gplbacktext{%
      \csname LTb\endcsname
      \put(624,503){\makebox(0,0)[r]{\strut{}$10^1$}}%
      \put(624,1175){\makebox(0,0)[r]{\strut{}$10^2$}}%
      \put(624,1846){\makebox(0,0)[r]{\strut{}$10^3$}}%
      \put(624,2518){\makebox(0,0)[r]{\strut{}$10^4$}}%
      \put(624,3189){\makebox(0,0)[r]{\strut{}$10^5$}}%
      \put(1428,283){\makebox(0,0){\strut{}$n = 10^3$}}%
      \put(2772,283){\makebox(0,0){\strut{}$n = 10^4$}}%
      \put(4115,283){\makebox(0,0){\strut{}$n = 10^5$}}%
    }%
    \gplgaddtomacro\gplfronttext{%
      \csname LTb\endcsname
      \put(151,1846){\rotatebox{-270}{\makebox(0,0){\strut{}required number of queries $m$}}}%
      \put(2771,63){\makebox(0,0){\strut{}number of agents $n$}}%
      \csname LTb\endcsname
      \put(1812,3016){\makebox(0,0)[r]{\strut{}$p = 0.1$}}%
      \csname LTb\endcsname
      \put(1812,2796){\makebox(0,0)[r]{\strut{}$p = 0.3$}}%
      \csname LTb\endcsname
      \put(1812,2576){\makebox(0,0)[r]{\strut{}$p = 0.5$}}%
      \csname LTb\endcsname
      \put(2771,3189){\makebox(0,0){\strut{}}}%
    }%
    \gplgaddtomacro\gplbacktext{%
      \csname LTb\endcsname
      \put(624,503){\makebox(0,0)[r]{\strut{} }}%
      \put(624,839){\makebox(0,0)[r]{\strut{} }}%
      \put(624,1175){\makebox(0,0)[r]{\strut{} }}%
      \put(624,1510){\makebox(0,0)[r]{\strut{} }}%
      \put(624,1846){\makebox(0,0)[r]{\strut{} }}%
      \put(624,2182){\makebox(0,0)[r]{\strut{} }}%
      \put(624,2518){\makebox(0,0)[r]{\strut{} }}%
      \put(624,2853){\makebox(0,0)[r]{\strut{} }}%
      \put(624,3189){\makebox(0,0)[r]{\strut{} }}%
      \put(756,283){\makebox(0,0){\strut{} }}%
      \put(1428,283){\makebox(0,0){\strut{} }}%
      \put(2100,283){\makebox(0,0){\strut{} }}%
      \put(2772,283){\makebox(0,0){\strut{} }}%
      \put(3443,283){\makebox(0,0){\strut{} }}%
      \put(4115,283){\makebox(0,0){\strut{} }}%
      \put(4787,283){\makebox(0,0){\strut{} }}%
    }%
    \gplgaddtomacro\gplfronttext{%
      \csname LTb\endcsname
      \put(525,1846){\rotatebox{-270}{\makebox(0,0){\strut{}}}}%
      \put(2771,-3){\makebox(0,0){\strut{}}}%
      \csname LTb\endcsname
      \put(4328,1336){\makebox(0,0)[r]{\strut{}$\lambda = 0$}}%
      \csname LTb\endcsname
      \put(4328,1116){\makebox(0,0)[r]{\strut{}$\lambda = 1$}}%
      \csname LTb\endcsname
      \put(4328,896){\makebox(0,0)[r]{\strut{}$\lambda = 2$}}%
      \csname LTb\endcsname
      \put(4328,676){\makebox(0,0)[r]{\strut{}$\lambda = 3$}}%
      \csname LTb\endcsname
      \put(2771,6378){\makebox(0,0){\strut{}}}%
    }%
    \gplbacktext
    \put(0,0){\includegraphics[width={252.00bp},height={167.90bp}]{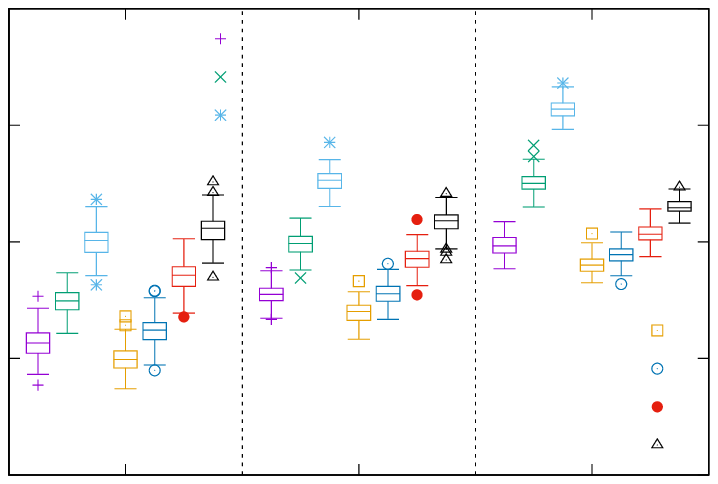}}%
    \gplfronttext
  \end{picture}%
\endgroup

%% file: figures/plot-amp.tex
\begingroup
  \makeatletter
  \providecommand\color[2][]{%
    \GenericError{(gnuplot) \space\space\space\@spaces}{%
      Package color not loaded in conjunction with
      terminal option `colourtext'%
    }{See the gnuplot documentation for explanation.%
    }{Either use 'blacktext' in gnuplot or load the package
      color.sty in LaTeX.}%
    \renewcommand\color[2][]{}%
  }%
  \providecommand\includegraphics[2][]{%
    \GenericError{(gnuplot) \space\space\space\@spaces}{%
      Package graphicx or graphics not loaded%
    }{See the gnuplot documentation for explanation.%
    }{The gnuplot epslatex terminal needs graphicx.sty or graphics.sty.}%
    \renewcommand\includegraphics[2][]{}%
  }%
  \providecommand\rotatebox[2]{#2}%
  \@ifundefined{ifGPcolor}{%
    \newif\ifGPcolor
    \GPcolorfalse
  }{}%
  \@ifundefined{ifGPblacktext}{%
    \newif\ifGPblacktext
    \GPblacktexttrue
  }{}%
  \let\gplgaddtomacro\g@addto@macro
  \gdef\gplbacktext{}%
  \gdef\gplfronttext{}%
  \makeatother
  \ifGPblacktext
    \def\colorrgb#1{}%
    \def\colorgray#1{}%
  \else
    \ifGPcolor
      \def\colorrgb#1{\color[rgb]{#1}}%
      \def\colorgray#1{\color[gray]{#1}}%
      \expandafter\def\csname LTw\endcsname{\color{white}}%
      \expandafter\def\csname LTb\endcsname{\color{black}}%
      \expandafter\def\csname LTa\endcsname{\color{black}}%
      \expandafter\def\csname LT0\endcsname{\color[rgb]{1,0,0}}%
      \expandafter\def\csname LT1\endcsname{\color[rgb]{0,1,0}}%
      \expandafter\def\csname LT2\endcsname{\color[rgb]{0,0,1}}%
      \expandafter\def\csname LT3\endcsname{\color[rgb]{1,0,1}}%
      \expandafter\def\csname LT4\endcsname{\color[rgb]{0,1,1}}%
      \expandafter\def\csname LT5\endcsname{\color[rgb]{1,1,0}}%
      \expandafter\def\csname LT6\endcsname{\color[rgb]{0,0,0}}%
      \expandafter\def\csname LT7\endcsname{\color[rgb]{1,0.3,0}}%
      \expandafter\def\csname LT8\endcsname{\color[rgb]{0.5,0.5,0.5}}%
    \else
      \def\colorrgb#1{\color{black}}%
      \def\colorgray#1{\color[gray]{#1}}%
      \expandafter\def\csname LTw\endcsname{\color{white}}%
      \expandafter\def\csname LTb\endcsname{\color{black}}%
      \expandafter\def\csname LTa\endcsname{\color{black}}%
      \expandafter\def\csname LT0\endcsname{\color{black}}%
      \expandafter\def\csname LT1\endcsname{\color{black}}%
      \expandafter\def\csname LT2\endcsname{\color{black}}%
      \expandafter\def\csname LT3\endcsname{\color{black}}%
      \expandafter\def\csname LT4\endcsname{\color{black}}%
      \expandafter\def\csname LT5\endcsname{\color{black}}%
      \expandafter\def\csname LT6\endcsname{\color{black}}%
      \expandafter\def\csname LT7\endcsname{\color{black}}%
      \expandafter\def\csname LT8\endcsname{\color{black}}%
    \fi
  \fi
    \setlength{\unitlength}{0.0500bp}%
    \ifx\gptboxheight\undefined%
      \newlength{\gptboxheight}%
      \newlength{\gptboxwidth}%
      \newsavebox{\gptboxtext}%
    \fi%
    \setlength{\fboxrule}{0.5pt}%
    \setlength{\fboxsep}{1pt}%
    \definecolor{tbcol}{rgb}{1,1,1}%
\begin{picture}(5040.00,3358.00)%
    \gplgaddtomacro\gplbacktext{%
      \csname LTb\endcsname
      \put(624,503){\makebox(0,0)[r]{\strut{}$0.0$}}%
      \put(624,1040){\makebox(0,0)[r]{\strut{}$0.2$}}%
      \put(624,1577){\makebox(0,0)[r]{\strut{}$0.4$}}%
      \put(624,2115){\makebox(0,0)[r]{\strut{}$0.6$}}%
      \put(624,2652){\makebox(0,0)[r]{\strut{}$0.8$}}%
      \put(624,3189){\makebox(0,0)[r]{\strut{}$1.0$}}%
      \put(756,283){\makebox(0,0){\strut{}$0$}}%
      \put(1428,283){\makebox(0,0){\strut{}$100$}}%
      \put(2100,283){\makebox(0,0){\strut{}$200$}}%
      \put(2772,283){\makebox(0,0){\strut{}$300$}}%
      \put(3443,283){\makebox(0,0){\strut{}$400$}}%
      \put(4115,283){\makebox(0,0){\strut{}$500$}}%
      \put(4787,283){\makebox(0,0){\strut{}$600$}}%
    }%
    \gplgaddtomacro\gplfronttext{%
      \csname LTb\endcsname
      \put(151,1846){\rotatebox{-270}{\makebox(0,0){\strut{}success rate}}}%
      \put(2771,63){\makebox(0,0){\strut{}number of queries $m$}}%
      \csname LTb\endcsname
      \put(4196,1556){\makebox(0,0)[r]{\strut{}AMP $0.1$}}%
      \csname LTb\endcsname
      \put(4196,1336){\makebox(0,0)[r]{\strut{}AMP $0.3$}}%
      \csname LTb\endcsname
      \put(4196,1116){\makebox(0,0)[r]{\strut{}greedy $0.1$}}%
      \csname LTb\endcsname
      \put(4196,896){\makebox(0,0)[r]{\strut{}greedy $0.3$}}%
      \csname LTb\endcsname
      \put(4196,676){\makebox(0,0)[r]{\strut{}greedy $0.3$}}%
      \csname LTb\endcsname
      \put(2771,3189){\makebox(0,0){\strut{}}}%
    }%
    \gplbacktext
    \put(0,0){\includegraphics[width={252.00bp},height={167.90bp}]{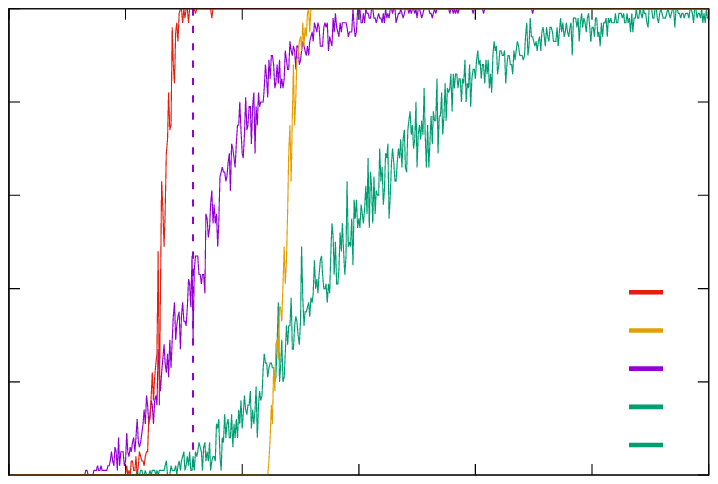}}%
    \gplfronttext
  \end{picture}%
\endgroup

%% file: figures/plot-overlap.tex
\begingroup
  \makeatletter
  \providecommand\color[2][]{%
    \GenericError{(gnuplot) \space\space\space\@spaces}{%
      Package color not loaded in conjunction with
      terminal option `colourtext'%
    }{See the gnuplot documentation for explanation.%
    }{Either use 'blacktext' in gnuplot or load the package
      color.sty in LaTeX.}%
    \renewcommand\color[2][]{}%
  }%
  \providecommand\includegraphics[2][]{%
    \GenericError{(gnuplot) \space\space\space\@spaces}{%
      Package graphicx or graphics not loaded%
    }{See the gnuplot documentation for explanation.%
    }{The gnuplot epslatex terminal needs graphicx.sty or graphics.sty.}%
    \renewcommand\includegraphics[2][]{}%
  }%
  \providecommand\rotatebox[2]{#2}%
  \@ifundefined{ifGPcolor}{%
    \newif\ifGPcolor
    \GPcolorfalse
  }{}%
  \@ifundefined{ifGPblacktext}{%
    \newif\ifGPblacktext
    \GPblacktexttrue
  }{}%
  \let\gplgaddtomacro\g@addto@macro
  \gdef\gplbacktext{}%
  \gdef\gplfronttext{}%
  \makeatother
  \ifGPblacktext
    \def\colorrgb#1{}%
    \def\colorgray#1{}%
  \else
    \ifGPcolor
      \def\colorrgb#1{\color[rgb]{#1}}%
      \def\colorgray#1{\color[gray]{#1}}%
      \expandafter\def\csname LTw\endcsname{\color{white}}%
      \expandafter\def\csname LTb\endcsname{\color{black}}%
      \expandafter\def\csname LTa\endcsname{\color{black}}%
      \expandafter\def\csname LT0\endcsname{\color[rgb]{1,0,0}}%
      \expandafter\def\csname LT1\endcsname{\color[rgb]{0,1,0}}%
      \expandafter\def\csname LT2\endcsname{\color[rgb]{0,0,1}}%
      \expandafter\def\csname LT3\endcsname{\color[rgb]{1,0,1}}%
      \expandafter\def\csname LT4\endcsname{\color[rgb]{0,1,1}}%
      \expandafter\def\csname LT5\endcsname{\color[rgb]{1,1,0}}%
      \expandafter\def\csname LT6\endcsname{\color[rgb]{0,0,0}}%
      \expandafter\def\csname LT7\endcsname{\color[rgb]{1,0.3,0}}%
      \expandafter\def\csname LT8\endcsname{\color[rgb]{0.5,0.5,0.5}}%
    \else
      \def\colorrgb#1{\color{black}}%
      \def\colorgray#1{\color[gray]{#1}}%
      \expandafter\def\csname LTw\endcsname{\color{white}}%
      \expandafter\def\csname LTb\endcsname{\color{black}}%
      \expandafter\def\csname LTa\endcsname{\color{black}}%
      \expandafter\def\csname LT0\endcsname{\color{black}}%
      \expandafter\def\csname LT1\endcsname{\color{black}}%
      \expandafter\def\csname LT2\endcsname{\color{black}}%
      \expandafter\def\csname LT3\endcsname{\color{black}}%
      \expandafter\def\csname LT4\endcsname{\color{black}}%
      \expandafter\def\csname LT5\endcsname{\color{black}}%
      \expandafter\def\csname LT6\endcsname{\color{black}}%
      \expandafter\def\csname LT7\endcsname{\color{black}}%
      \expandafter\def\csname LT8\endcsname{\color{black}}%
    \fi
  \fi
    \setlength{\unitlength}{0.0500bp}%
    \ifx\gptboxheight\undefined%
      \newlength{\gptboxheight}%
      \newlength{\gptboxwidth}%
      \newsavebox{\gptboxtext}%
    \fi%
    \setlength{\fboxrule}{0.5pt}%
    \setlength{\fboxsep}{1pt}%
    \definecolor{tbcol}{rgb}{1,1,1}%
\begin{picture}(5040.00,3358.00)%
    \gplgaddtomacro\gplbacktext{%
      \csname LTb\endcsname
      \put(624,503){\makebox(0,0)[r]{\strut{}$0.0$}}%
      \put(624,1040){\makebox(0,0)[r]{\strut{}$0.2$}}%
      \put(624,1577){\makebox(0,0)[r]{\strut{}$0.4$}}%
      \put(624,2115){\makebox(0,0)[r]{\strut{}$0.6$}}%
      \put(624,2652){\makebox(0,0)[r]{\strut{}$0.8$}}%
      \put(624,3189){\makebox(0,0)[r]{\strut{}$1.0$}}%
      \put(756,283){\makebox(0,0){\strut{}$0$}}%
      \put(1428,283){\makebox(0,0){\strut{}$100$}}%
      \put(2100,283){\makebox(0,0){\strut{}$200$}}%
      \put(2772,283){\makebox(0,0){\strut{}$300$}}%
      \put(3443,283){\makebox(0,0){\strut{}$400$}}%
      \put(4115,283){\makebox(0,0){\strut{}$500$}}%
      \put(4787,283){\makebox(0,0){\strut{}$600$}}%
    }%
    \gplgaddtomacro\gplfronttext{%
      \csname LTb\endcsname
      \put(151,1846){\rotatebox{-270}{\makebox(0,0){\strut{}overlap}}}%
      \put(2771,63){\makebox(0,0){\strut{}number of queries $m$}}%
      \csname LTb\endcsname
      \put(4196,896){\makebox(0,0)[r]{\strut{}$p = 0.1$}}%
      \csname LTb\endcsname
      \put(4196,676){\makebox(0,0)[r]{\strut{}$p = 0.3$}}%
      \csname LTb\endcsname
      \put(2771,6378){\makebox(0,0){\strut{}}}%
    }%
    \gplbacktext
    \put(0,0){\includegraphics[width={252.00bp},height={167.90bp}]{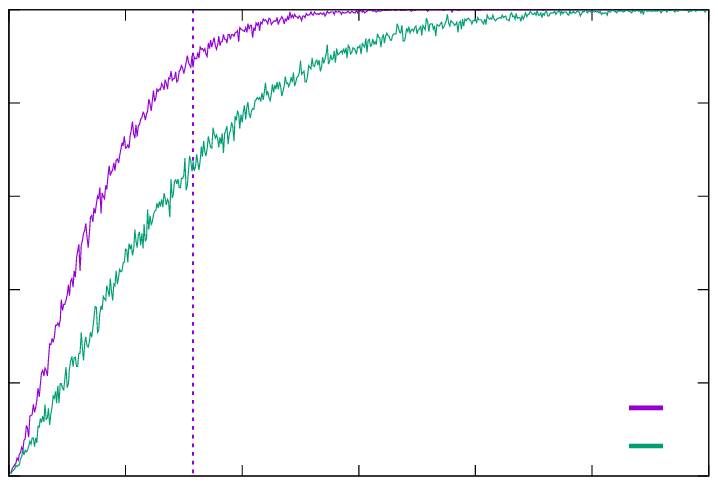}}%
    \gplfronttext
  \end{picture}%
\endgroup